\documentclass{elsarticle}
\usepackage{lineno,hyperref}
\usepackage{amsmath,amssymb,amsthm}
\usepackage{bm}
\usepackage{color}
\modulolinenumbers[5]

\journal{Journal of \LaTeX\ Templates}

\newtheorem{theorem}{Theorem}[section]
\newtheorem{define}[theorem]{Definition}
\newtheorem{remark}[theorem]{Remark}
\newtheorem{prop}[theorem]{Proposition}
\newtheorem{corollary}[theorem]{Corollary}
\newtheorem{lemma}[theorem]{Lemma}
\newtheorem{example}[theorem]{Example}

\newtheorem{notation}[theorem]{Notation}

\newcommand{\Z} {\mathbb{Z}}
\newcommand{\N} {\mathbb{N}}
\newcommand{\bP} {\mathbb{P}}
\newcommand{\bV} {\mathbb{V}}
\newcommand{\tdeg}{{\rm tdeg}}
\newcommand{\bfm}{{\bf m}}
\newcommand{\bfn}{{\bf n}}
\newcommand{\bfc}{{\bf c}}
\newcommand{\bfd}{{\bf d}}
\newcommand{\bfa}{{\bf a}}
\newcommand{\bfb}{{\bf b}}
\newcommand{\bfq}{{\bf q}}
\newcommand{\bfy}{{\bf y}}
\newcommand{\frakp}{{\mathfrak p}}
\newcommand{\frakq}{{\mathfrak q}}
\newcommand{\calB}{{\mathcal B}}
\newcommand{\calS}{{\mathcal S}}
\newcommand{\calC}{{\mathcal C}}

\newcommand{\calL}{{\mathcal L}}

\newcommand{\calQ}{{\mathcal Q}}
\newcommand{\calR}{{\mathcal R}}
\newcommand{\frakP}{{\mathfrak P}}
\newcommand{\frakQ}{{\mathfrak Q}}
\newcommand{\frakL}{{\mathfrak L}}

\newcommand{\bfxi}{{\bm \xi}}
\newcommand{\ord}{{\rm ord}}
\newcommand{\divs}{{\rm div}}
\newcommand{\diag}{{\rm diag}}
\newcommand{\supp}{{\rm supp}}
\newcommand{\rank}{{\rm rank}}
\newcommand{\res}{{\rm res}}
\newcommand{\GL}{{\rm GL}}
\newcommand{\msindex}{{\rm m.s.index}}
\newcommand{\overkt}{{\overline{k(t)}}}

\begin{document}

\begin{frontmatter}

\title{Rational Solutions of First Order Algebraic Ordinary Differential Equations}
\tnotetext[mytitlenote]{This work was supported by NSFC under Grants No.11771433 and No.11688101, and by Beijing Natural Science Foundation (Z190004).}

\author{Ruyong Feng and Shuang Feng}
\address{KLMM,Academy of Mathematics and Systems Science, Chinese Academy of Sciences and School of Mathematics, University of Chinese Academy of Sciences, 100190, Beijing\\
ryfeng@amss.ac.cn, fengshuang15@mails.ucas.ac.cn}

\begin{abstract}
Let $f(t, y,y')=\sum_{i=0}^d a_i(t, y)y'^i=0$ be a first order ordinary differential equation with polynomial coefficients.
Eremenko in 1999 proved that there exists  a constant $C$ such that every rational solution of $f(t, y,y')=0$ is of degree not greater than $C$.  Examples show that this degree bound $C$ depends not only on  the degrees of $f$ in $t,y,y'$ but also on the coefficients of $f$ viewed as polynomial in $t,y,y'$.  In this paper, we show that if
$$
     \max_{i=0}^d \{ \deg(a_i,y)-2(d-i)\}>0
$$
then the degree bound $C$ only depends on the degrees of $f$, and furthermore we present an explicit expression for $C$ in terms of the degrees of $f$.
\end{abstract}

\begin{keyword}
\texttt{first order AODE, rational solution, degree bound, height}
\MSC[2010] 34A05\sep  68W30
\end{keyword}

\end{frontmatter}

\section{Introduction}
The study of first order algebraic ordinary differential equations (AODEs in short) has a long history, which can be at least tracked back to the time of Fuchs and Poincar\'e.  Fuchs presented a sufficient and necessary condition so called Fuchs' criterion for a first order AODE having no movable singularities. Roughly speaking, an AODE is said to have movable singularities if it has a solution (with arbitrary constants) whose branch points depend on arbitrary constants. For instance the solution $y=\sqrt{t+c}$ of $2yy'-1=0$ has branch points $t=-c$, where $c$ is an arbitrary constant, so $2yy'-1=0$ has movable singularities.
Based on differential algebra developed by Ritt \cite{ritt} and the theory of algebraic function field of one variable, Matsuda \cite{matsuda} reproduced many classic results of first order AODEs. In particular, he presented an algebraic definition of movable singularities. In 1999, combining Matsuda's results and height estimates of points on plane algebraic curves, Eremenko showed that rational solutions of  first order AODEs have bounded degrees. In \cite{feng-feng}, we proved that if a first order AODE has movable singularities then it has only finitely many rational solutions. As for algebraic solutions of first order AODEs, Freitag and Moosa \cite{freitag-moosa} showed that  they are of bounded heights.

On the other hand, the algorithmic aspects of computing closed form solutions of AODEs have been extensively studied in the past decades. Several algorithms have been developed for computing closed form solutions (e.g. liouvillian solutions) of linear homogeneous differential equations (see \cite{barkatou,kovacic,vanhoeij-ragot-ulmer-weil,singer1,vanderput-singer} etc). Yet, the situation is different in the nonlinear case. Existing algorithms were only valid for AODEs of special types. Based on parametrization of algebraic curves, Aroca et al \cite{aroca-cano-feng-gao,feng-gao} gave two complete methods for finding rational and algebraic solutions of first order autonomous AODEs. Their methods were generalized by Winkler and his collegues to the class of first order non-autonomous AODEs whose rational general solutions involve arbitrary constants rationally as well as some other certain classes of AODEs (see \cite{vo-grasegger-winkler1,vo-grasegger-winkler2,chau-winkler, winkler} etc). Particularly, in \cite{vo-grasegger-winkler1}, the authors introduced a class of first order AODEs called maximally comparable AODEs and presented an algorithm to compute a degree bound for rational solutions of this kind of equations as well as first order quasi-linear AODEs.   Readers are referred to \cite{winkler} for a survey of recent developments in this direction.
Theoretically, it suffices to compute a degree bound for all rational solutions of  a first order AODE to find all its rational solutions. The following example implies that the degrees of rational solutions may depend not only on the degrees of the original equation but also on its constant coefficients.
\begin{example}
Let $n$ be an integer. Then $y=t^n$ is a rational solution of $ty'-ny=0$.  The degree of $t^n$ depends on the constant coefficient $n$ of $ty'-ny$.
\end{example}
Let $f=\sum_{i=0}^d a_i(t,y)y'^i=0$ be an irreducible first order AODE.  Set
\begin{equation*}
\label{eqn:index}
     \msindex(f)=\max_{i=0}^d \{ \deg(a_i,y)-2(d-i)\}.
\end{equation*}
Fuchs' criterion (see Remark on page 14 of \cite{matsuda}) implies that $f=0$ has movable singularities if $\msindex(f)>0$. On the other hand, it was proved in \cite{eremenko} that if $f=0$ has movable singularities then it can be transferred into an AODE $g$ with positive $\msindex$. This motivates us to focus on first order AODEs with positive $\msindex$.
We prove that for an irreducible first order AODE $f=0$ with $\msindex(f)>0$ the degrees of rational solutions of $f=0$ are independent of the constant coefficients of $f$ and furthermore we present an explicit degree bound in terms of the degrees of $f$. The key step to obtain this degree bound is to estimate the heights of points on plane algebraic curves. This height estimate is a special case of the result about heights on complete nonsingular varieties (see for instance Proposition 3 on page 89 of \cite{lang}). Eremenko in \cite{eremenko} provided a simple proof for this special case based on the Riemann-Roch Theorem. We follow Eremenko's proof but present explicit bounds for each step.

The paper is organized as follows. In Section 2, we introduce some basic materials used in the later sections. In Sections 3, we estimate the degrees and heights for elements in a Riemann-Roth space. In Section 4, we present an explicit bound for the heights of points on a plane algebraic curve. Finally, in Section 5, we apply the results in Section 4 to first order AODEs.

Throughout this paper, $\Z$ stands for the ring of integers, $k, K$ for algebraically closed fields of characteristic zero, and $R$ and $\calR$ for algebraic function fields over $k$ and $K$ respectively.  $\bP^m(\cdot)$ denotes the projective space of dimension $m$ over a field and $\bV(\cdot)$ denotes the variety in a projective space defined by a set of homogeneous polynomials.

\section{Basic materials}
In this section, we will introduce some basic materials used in this paper, including differential rings, algebraic function fields of one variable and heights. Readers are referred to \cite{ritt, matsuda, chevalley, lang} for details.
\subsection{Differential fields associated to  AODEs}
In this subsection, we introduce some basic notations of differential algebra.
\begin{define}
A derivation on a ring $\calR$ is a map $\delta: \calR \rightarrow \calR$ satisfying that for all $a,b\in \calR$,
$$
   \delta(a+b)=\delta(a)+\delta(b),\,\,\delta(ab)=\delta(a)b+a\delta(b).
$$
A ring (resp. field) equipped with a derivation is called a differential ring (resp. differential field).  An ideal $I\subset \calR$ is called a differential ideal if $\delta(I)\subset I$.
\end{define}
The field $k(t)$ of rational functions in $t$ can be endowed with a structure of differential field whose derivation $\delta$ is the usual derivation with respect to $t$, i.e. $\delta=\frac{\rm d}{{\rm d} t}$. Set $y_0=y$ and denote
$$k(t)\{y\}=k(t)[y_0,y_1,\dots]$$
 where $y_0,y_1,\dots$ are indeterminates.  One can extend the derivation $\delta$ on $k(t)$ to a derivation $\delta'$ on $k(t)\{y\}$ by assigning $y_i=\delta'^i(y_0)$ so that $k(t)\{y\}$ becomes a differential ring.  For the sake of notations, we use $\delta$ in place of $\delta'$. Elements in $k(t)\{y\}$ are called differential polynomials over $k(t)$. Let $f$ be a differential polynomial not in $k(t)$. Then there is a unique integer $d$ such that
$f\in k(t)[y_0,\dots,y_d]\setminus k(t)[y_0,\dots,y_{d-1}]$. This integer is called the order of $f$. We shall use $[\cdot]$ (resp. $\langle \cdot \rangle$) to stand for the differential (resp. algebraic) ideal generated by a set of differential polynomials (resp. polynomials) respectively. Suppose that $f$ is irreducible viewed as an algebraic polynomial.  Set
 $$
   \Sigma_f=\left\{ A\in k(t)\{y\} | \,\exists\, m>0\, \,\mbox{s.t.}\, S^m A^m \in [f] \right\}
$$
where $S=\partial f/\partial y_d$ and $d$ is the order of $f$. It was proved on page 30 of \cite{ritt} that $\Sigma_f$ is a prime differential ideal and so $k(t)\{y\}/\Sigma_f$ is a differential domain.
Lemma 2.2 of \cite{feng-feng} implies that the field of  fractions of $k(t)\{y\}/\Sigma_f$ is isomorphic to that of $k(t)[y_0,\dots,y_d]/\langle f \rangle$. Under this isomorphism, the field of fractions of $k(t)[y_0,\dots,y_d]/\langle f\rangle$ can be endowed with a structure of differential field. We shall still use $\delta$, or $'$ in short,  to denote the induced derivation on the field of fractions of $k(t)[y_0,\dots,y_d]/\langle f \rangle$.

In this paper, the first order AODEs under consideration are differential equations of the following form
\begin{equation}
\label{eq:differentialeqn}
 f(y,y')=0
\end{equation}
where $f(y,y')\in k(t)[y,y']\setminus k(t)$.
\begin{define}
An element $r(t)\in k(t)$ satisfying
$f(r(t),r'(t))=0$ is called a rational solution of $f(y,y')=0$.
\end{define}
Remark that the derivation $\delta$ in $k(t)$ can be uniquely extended to a derivation in $\overkt$ which we shall still denote by $\delta$. Assume that viewed as a polynomial in $\overkt[y,y']$, $f$ is irreducible over $\overline{k(t)}$. Then the field of fractions of $\overkt[y,y']/\langle f(y,y') \rangle$ is not only an algebraic function field over $\overkt$ but also a differential field.

\subsection{Algebraic function fields of one variable}
Let $K$ be an algebraically closed field of characteristic zero and  $\calR$ an extension field of $K$.
We say $\calR$ is an {\em algebraic function field of one variable} over $K$ if $\calR$ satisfies the following conditions: there is an element $a$ of $\calR$ which is transcendental over $K$, and $\calR$ is algebraic of finite degree over $K(a)$. Assume $\calR$ is an algebraic function field of one variable over $K$. A {\em valuation ring} of $\calR$ over $K$ is a subring $V$ satisfying that
\begin{enumerate}
\item $K\subset V\neq \calR$; and
\item if $a\in \calR \setminus V$, then $a^{-1}\in V$.
\end{enumerate}
All non-invertible elements of $V$ form a maximal ideal $\frakP$ which is called a place of $\calR$, and $V$  is called the corresponding valuation ring of $\frakP$. Let $V$ be a valuation ring with $\frakP$ as place. There is an element $u\in V$, called a {\em local uniformizer} of $\frakP$ or $V$, such that $\frakP=uV$ and $\bigcap_{n=1}^\infty u^n V=\{0\}$. The factor ring $V/\frakP$ is equal to $K$ since $K$ is algebraically closed. For every valuation ring $V$ with place $\frakP$, we define a map
\[
\pi_{\frakP}: \calR \longrightarrow K\cup \{\infty\}
\]
satisfying if $a\in V$ then $\pi_{\frakP}(a)=a+\frakP\in V/\frakP=K$, otherwise $\pi_{\frakP}(a)=\infty$. It is well-known that $\calR$ admits infinitely many places, and there is one-to-one correspondence between places and valuation rings.

Let $\frakP$ be a place of $\calR$ and $V$ the corresponding valuation ring of $\frakP$. Let $u$ be a local uniformizer of $\frakP$. Then for every non-zero element $a$ of $\calR$, there is a unique integer $n$ such that
\[
a=u^nv
\]
for some invertible element $v\in V$. It is easy to see that the integer $n$ is independent of the choice of local uniformizers.
Such $n$ is called the {\em order} of $a$ at $\frakP$ and denoted by $\nu_{\frakP}(a)$. We make the convention to write $\nu_{\frakP}(0)=\infty$. Then the place $\frakP$ induces a map $\nu_{\frakP}$ from $\calR$ to $\Z$ sending $a$ to $\nu_{\frakP}(a)$. This map $\nu_{\frakP}$ is called the {\em order function} at $\frakP$. For $a,b\in \calR$, we have
$$
\nu_{\frakP}(ab)=\nu_{\frakP}(a)+\nu_{\frakP}(b),\,\,\nu_{\frakP}(a+b)\geq \min\{\nu_{\frakP}(a),\nu_{\frakP}(b)\}
$$
where the equality in the later formula holds if $\nu_{\frakP}(a)\neq \nu_{\frakP}(b)$. Let $a\in \calR$ and $\frakP$ be a place. We say $\frakP$ is a {\em zero} of $a$ if $\nu_{\frakP}(a)>0$, and a {\em pole} of $a$ if $\nu_{\frakP}(a)<0$. Every non-zero element of $\calR$ admits only finitely many zeros and poles.

A {\em divisor} in $\calR$ is  a formal sum
$$
D=\sum_{\frakP} n_{\frakP} \frakP
$$
for all the places of $\calR$, where $n_{\frakP} \in \Z$ and $n_{\frakP}=0$ for all but finitely many $\frakP$. It is easy to see that the set of divisors in $\calR$ forms an abelian group.
$D$ is {\em effective} if $n_{\frakP}\geq 0$ for all $\frakP$.
The {\em degree} of $D$, denoted by $\deg(D)$, is defined to be $\sum n_{\frakP}$ and the {\em support} of $D$, denoted by $\supp(D)$, is defined to be $\{\frakP \,|\, n_{\frakP}\neq 0\}$. For brief, we denote
$$
D^{+}=\sum_{n_\frakP>0}n_{\frakP} \frakP,\quad D^{-}=\sum_{n_{\frakP}<0}-n_{\frakP} \frakP.
$$
Let $D_1=\sum_{\frakP} n_{\frakP} \frakP$ and $D_2=\sum_{\frakP} m_{\frakP} \frakP$ be two divisors in $\calR$, we write $D_1\geq D_2 $ provided $D_1-D_2$ is effective. For every non-zero element $a$ of $\calR$, we denote
$$
 \divs(a)=\sum_{\frakP} \nu_{\frakP}(a) \frakP
$$
where $\frakP$ ranges over all places of $\calR$. Then $\divs(a)$ is a divisor of degree $0$. For a divisor $D$, we denote
$$
    \frakL(D)=\{a\in \calR\mid \divs(a)+D\geq 0\}\cup \{0\},
$$
which is called the Riemann-Roch space of $D$.
It is well-known that each Riemann-Roch space is a $K$-vector space of finite dimension.
The Riemann-Roch Theorem implies that if $D$ is a divisor whose degree is not less than the genus of $\calR$ then $\frakL(D)$ is of positive dimension.

Let $f\in K[x_0,x_1]\setminus K$ be irreducible.  One sees that  the field of fractions of $K[x_0,x_1]/\langle f \rangle$ is an algebraical function field of one variable over $K$ which is called the algebraic function field of $f$. For an irreducible homogeneous polynomial $F$ in $K[x_0,x_1,x_2]$, the corresponding algebraic function field is defined to be the algebraic function field of $F(x_0,x_1,1)$.  Remark that the algebraic function fields of $F(1,x_1,x_2), F(x_0,1,x_2)$ and $F(x_0,x_1,1)$ are all isomorphic. 

\subsection{Models of algebraic function fields of one variable}
Let $\calR$ be an algebraic function field of one variable over $K$. The set of all places of $\calR$ can be viewed as a nonsingular model of $\calR$. On the other hand, let $F$ be an irreducible homogeneous polynomial $F\in K[x_0,x_1,x_2]$ whose algebraic function field is $\calR$. Then the projective curve $F=0$ is another model of $\calR$. There is a surjective map from a nonsingular model of $\calR$  to the curve $F=0$. To describe this map precisely, let $\xi_0,\xi_1,\xi_2$ be three nonzero elements of $\calR$ satisfying that
$$\calR=K(\xi_0/\xi_2,\xi_1/\xi_2)\,\,\mbox{and}\, \,F(\xi_0,\xi_1,\xi_2)=0.$$
Set $\bfxi=(\xi_0,\xi_1,\xi_2)$. Let $\frakP$ be a place of $\calR$ with $u$ as local uniformizer. Denote by $\ell=\min_{i} \{\nu_{\frakP}(\xi_i)\}$. One sees that $\nu_{\frakP}(u^{-\ell}\xi_i)\geq 0$ and moreover not all $\pi_{\frakP}(u^{-\ell}\xi_i)$ are zero. Therefore $(\pi_{\frakP}(u^{-\ell}\xi_0), \pi_{\frakP}(u^{-\ell}\xi_1),\pi_{\frakP}(u^{-\ell}\xi_2))$ defines a point of $\bP^2(K)$. Remark that this point does not depend on the choice of $u$.
\begin{define}
 We call $(\pi_{\frakP}(u^{-\ell}\xi_0), \pi_{\frakP}(u^{-\ell}\xi_1),\pi_{\frakP}(u^{-\ell}\xi_2))$ the center of $\frakP$ with respect to $\bfxi$. Denote by $\calC(\bfxi)$ the set of centers with respect to $\bfxi$.
\end{define}
 We claim that $\calC(\bfxi)$ is the plane projective curve in $\bP^2(K)$ defined by $F$ and  the map sending $\frakP$ to the center of $\frakP$ with respect to $\bfxi$ is the required map. It is easy to verify that $F(\bfc)=0$ for all $\bfc\in \calC(\bfxi)$. Conversely, let $(c_0,c_1,c_2)$ be a point of $F=0$. Without loss of generality, we may assume that $c_0\neq 0$.  Then $F(1,c_1/c_0,c_2/c_0)=0$. Remark tha $\calR=K(\xi_1/\xi_0,\xi_2/\xi_0)$.  As $F(1,\xi_1/\xi_0,\xi_2/\xi_0)=0$, due to Corollary 2 on page 8 of \cite{chevalley}, there is a place $\frakP$ containing $\xi_1/\xi_0-c_1/c_0$ and $\xi_2/\xi_0-c_2/c_0$.  For this place, one has that
$\nu_{\frakP}(\xi_1)\geq \nu_{\frakP}(\xi_0), \nu_{\frakP}(\xi_2)\geq \nu_{\frakP}(\xi_0) $ and furthermore $\pi_{\frakP}(\xi_i/\xi_0)=c_i/c_0$. Write $\ell=\nu_{\frakP}(\xi_0)$. Then the center of $\frakP$ with respect to $\bfxi$ is
 \begin{align*}
    (\pi_{\frakP}(u^{-\ell}\xi_0),\pi_{\frakP}(u^{-\ell}\xi_1),\pi_{\frakP}(u^{-\ell}\xi_2))=\pi_{\frakP}(u^{-\ell}\xi_0)(1,c_1/c_0,c_2/c_0).
\end{align*}
 This implies that $(c_0,c_1,c_2)\in \calC(\bfxi)$.
\begin{define}
We call $\calC(\bfxi)$ or $F=0$ a plane projective model of $\calR$.
\end{define}
The plane projective models of $\calR$ usually have singularities.  Let $\calC$ be an irreducible projective curve in $\bP^2(K)$ defined by a homogeneous polynomial $F$.  A point $\bfc$ of $\calC$ is said to be of multiplicity $r$, if  all derivatives of $F$ up to and including the $(r-1)$-th  vanish at $\bfc$ but not all the $r$-th derivatives vanish at $\bfc$. Suppose that $\bfc$ is a point of $\calC$ with multiplicity $r$. If $r=1$, then $\bfc$ is called a simple point of $\calC$, otherwise a singular point of $\calC$. A point of multiplicity $r$ is ordinary if the $r$ tangents to $\calC$ at this point are distinct, otherwise it is non-ordinary. Due to Propositon on page of \cite{fulton}, $\calR$ has always a plane projective model with only ordinary singularities.

Let $\Phi=(\phi_0,\phi_1,\phi_2)$ be an invertible transformation, where $\phi_0,\phi_1,\phi_2$ are homogeneous polynomials in $K[x_0,x_1,x_2]$ of the same degree and they have no common factors.  We further assume that $\phi_i(\bfxi)\neq 0$ for all $i=0,1,2$. Then
$$\calR=K\left(\frac{\phi_0(\bfxi)}{\phi_2(\bfxi)},\frac{\phi_1(\bfxi)}{\phi_2(\bfxi)} \right).$$
\begin{prop}
\label{prop:centertransformation}
Let $\Phi,\bfxi$ be as above and $\frakP$ a place of $\calR$. Assume that $\bfc$ is the center of $\frakP$ with respect to $\bfxi$. If $\Phi(\bfc)\neq (0,0,0)$, then $\Phi(\bfc)$ is the center of $\frakP$ with respect to $\Phi(\bfxi)$.
\end{prop}
 \begin{proof}
 Let $u$ be a local uniformizer of $\frakP$ and $\ell=\min_{i=0}^2 \{\nu_{\frakP}(\xi_i)\}$.
 One has that
 $$(\pi_{\frakP}(u^{-\ell}\xi_0),\pi_{\frakP}(u^{-\ell}\xi_1),\pi_{\frakP}(u^{-\ell}\xi_2))=\lambda \bfc$$
 for some nonzero $\lambda\in K$. Denote $m=\tdeg(\phi_i)$. Then
 $$\pi_{\frakP}(\Phi(u^{-\ell}\bfxi))=\Phi(\pi_\frakP(u^{-\ell}\bfxi))=\Phi(\lambda \bfc)=\lambda^m \Phi(\bfc)\neq (0,0,0).$$
 This implies that
 $\min_{i=0}^2 \{\nu_{\frakP}(\phi_i(u^{-\ell}\bfxi))\}=0.$
In other words,
 $$
     \min_{i=0}^2 \{ \nu_{\frakP}(\phi_i(\bfxi))\}=m\ell.
 $$
 Then the center of $\frakP$ with respect to $\Phi(\bfxi)$ is
 $$
     \pi_{\frakP}(u^{-m\ell}\Phi(\bfxi))=\pi_{\frakP}(\Phi(u^{-\ell}\bfxi))=\Phi(\lambda \bfc)=\lambda^m \Phi(\bfc).
 $$
 Hence $\Phi(\bfc)$ is the center of $\frakP$ with respect to $\Phi(\bfxi)$.
 \end{proof}

\subsection{Heights}
 All algebraic function fields under consideration in this subsection are finite extensions of $k(t)$. They are algebraic function fields of one variable over $k$ and the places and order functions in them are defined as the same as in the previous subsection. 

\begin{define}
\label{def:height}
All points are considered as points in some suitable projective spaces over $\overkt$.
\begin{enumerate}
\item
Given $\bfa=(a_0,\dots,a_m)\in \bP^m(\overkt)$, let $R$ be a finite extension of $k(t)$ containing all $a_i$. We define the {\em height} of $\bfa$, denoted by $T(\bfa)$,  to be
$$
  \frac{\sum_{\frakp} \max_{i=0}^m\{-\nu_{\frakp}(a_i)\}}{[R:k(t)]}
$$
where $\frakp$ ranges over all places of $R$.
\item For $A=(a_{i,j})\in \GL_3(\overkt)$,  we define
$$
    T(A)=T((a_{1,1},a_{1,2}, a_{1,3},\dots,a_{3,3})).
$$
\item For $a\in \overkt$, we define the height of $a$ to be $T((1,a))$, denoted by $T(a)$.
\item Let $F$ be a polynomial in $\overkt[x_0,\dots,x_m]$. Suppose that $F$ contains at least two terms. We define the height of $F$, denoted by $T(F)$, to be $T(\bfc)$ where $\bfc$ is the point in a suitable projective space formed by the coefficients of $F$. For convention, when $F$ only contains one term, we defined $T(F)$ to be zero.
\item Let $V$ be a hypersurface in $\bP^m(\overkt)$ defined by $F\in \overkt[x_0,\dots,x_m]$. We define the height of $V$, denoted by $T(V)$, to be $T(F)$.
\end{enumerate}
\end{define}

\begin{remark}
\label{rem:heights}
Assume that $\bfa=(a_0,\dots,a_m), \bfb=(b_0,\dots,b_m)\in \bP^m(\overkt)$.
\begin{enumerate}
\item
 One sees that $T(\bfa)$ is independent of the choice of homogeneous coordinates and the choice of $R$. Without loss of generality, we suppose $a_0=1$, then
 $$
 T(\bfa)=\frac{\sum_{\frakp} \max\{0,-\nu_\frakp(a_1),\dots,\nu_\frakp(a_m)\}}{[R:k(t)]}\geq 0.
 $$
\item Assume $R$ is a finite extension of $k(t)$ containing all $a_i$ and $b_i$. Then one sees that if
$\max_i\{-\nu_\frakp(a_i)\}\geq \max_i\{-\nu_{\frakp}(b_i)\}$ for all places $\frakp$ of $R$ then $T(\bfa)\geq T(\bfb)$.
\item Suppose that $a_0,a_1,\dots,a_m\in k[t]$  and $\gcd(a_0,\dots,a_m)=1$. Then
$$
T(\bfa)=\max \{\deg(a_0),\dots,\deg(a_m)\}.
$$
To see this, let $R=k(t)$. Then
$$
T(\bfa)=\sum_{\frakp} \max \{-\nu_{\frakp}(a_0),\dots,-\nu_\frakp(a_m)\}
$$
where $\frakp$ ranges over all the places of $R$. Note that every place of $R$ has a local uniformizer of the form $1/t$ or $t-c$ for some $c\in k$.  Suppose the place $\frakp$ has $t-c$ as a local uniformizer. Then $\nu_\frakp(a_i)>0$ if and only if $(t-c )| a_i$. Since $\gcd(a_0,\dots,a_m)=1$, there is some $i_0$ such that $\nu_\frakp(a_{i_0})=0$. This implies that
for places $\frakp$ with $t-c, c\in k$ as local uniformizers,
$$\max \{-\nu_\frakp(a_0),\dots,-\nu_\frakp(a_m)\}=0.$$
For the place with $1/t$ as local uniformizer, one has that $\nu_\frakp(a_i)=-\deg(a_i)$. So for this place,
$$\max \{-\nu_\frakp(a_0),\dots,-\nu_\frakp(a_m)\}=\max\{\deg(a_0),\dots,\deg(a_m)\}.$$
Consequently, $T(\bfa)=\max\{\deg(a_0),\dots,\deg(a_m)\}$.
\item
Let $a\in \overkt$ and $R=k(t,a)$. Let $g(t,x)$ be a nonzero irreducible polynomial over $k$ such that $g(t,a)=0$. It is clear that $T(a)=0$ if $a\in k$. Assume that $a\notin k$ and $\frakp_1,\dots,\frakp_s$ are all distinct poles of $a$ in $R$, then
\[
T(a)=\frac{-\sum_{i=1}^{s}\nu_{\frakp_i}(a)}{[R:k(t)]}=\frac{[R:k(a)]}{[R:k(t)]}=\frac{\deg(g,t)}{\deg(g,x)}.
\]
In particular, if $a\in k(t)$ then $T(a)=\deg(a)$ which is defined to be the maximun of the degrees of the denominator and numerator of $a$.
\end{enumerate}
\end{remark}

From the above remark, it is easy to see that for $a\in \overkt\setminus \{0\}$ and $i\in \Z$
$$
 T(a^i)=|i| T(a).
$$
\begin{prop}\label{prop:heightproperty}
Let $a,b\in \overkt$, $c_1,\dots,c_4\in k$ with $c_1c_4-c_2c_3\neq 0$. Then
\begin{enumerate}
\item $T\left(\frac{c_1a+c_2}{c_3a+c_4}\right)=T(a)$ if $c_3a+c_4\neq 0$;
\item $T(ab)\leq T(a)+T(b)$;
\item $T(a+\lambda b)\leq T(a)+T(b)$ for all $\lambda\in k$.
\end{enumerate}
\end{prop}

\begin{proof}
$1.$ If $a\in k$ then the assertion is obvious. Suppose that $a\notin k$.  Let $R=k(t,a)$. Then $R=k(t, (c_1a+c_2)/(c_3a+c_4))$. The assertion follows from Remark~\ref{rem:heights} and the fact that  $$\left[R:k\left(\frac{c_1a+c_2}{c_3a+c_4}\right)\right]=[R:k(a)].$$

$2.$ Let $R=k(t,a,b)$. For each place $\frakp$ of $R$,
    $
    -\nu_{\frakp}(ab)=-\nu_{\frakp}(a)-\nu_{\frakp}(b)
    $
    and thus
    $$
    \max\{0,-\nu_{\frakp}(ab)\}\leq \max\{0,-\nu_{\frakp}(a)\}+\max\{0,-\nu_{\frakp}(b)\}.
    $$
   The assertion then follows from Remark~\ref{rem:heights}.

$3.$ Use an argument similar to that in 2. and the fact that
    $$
    -\nu_{\frakp}(a+\lambda b)\leq \max\{-\nu_{\frakp}(a),-\nu_{\frakp}(b)\}.
    $$
  \end{proof}
The following examples show that the equalities may hold in $2$ and $3$ of Proposition~\ref{prop:heightproperty}.
\begin{example}
Let $a=t^2, b=t^3+1$ and $\lambda\in k\setminus \{0\}$. Then
$$T(ab)=5=T(a)+T(b),\,\,T(1/a+\lambda/b)=5=T(1/a)+T(1/b.)$$
Moreover both of them are greater than the maximun of $T(a), T(b)$.
\end{example}

In the following, $\bfy$ stands for the vector with indeterminates $y_1,\dots,y_s$ and $\bfy^\bfd$ denotes $\prod_{i=1}^s y_i^{d_i}$ for $\bfd=(d_1,\dots,d_s)\in \Z^s$.
\begin{prop}\label{prop:height2}
Let $f$ and $g$ be polynomials in $\overkt[\bfy]$, then
\begin{enumerate}
\item $T(fg)\leq T(f)+T(g);$
\item If both $f$ and $g$ have 1 as a coefficient, then  $T(f+g)\leq T(f)+T(g).$
\end{enumerate}
\end{prop}
\begin{proof}
Write $f=\sum_{\bfd} a_{\bfd} \bfy^\bfd$ and $g=\sum_{\bfd} b_\bfd \bfy^\bfd$ with $a_\bfd,b_\bfd\in\overkt$. Let $R$ be a finite extension of $k(t)$ containing all $a_\bfd, b_\bfd$, and $\frakp$ a place of $R$.

1. Each coefficient of $fg$ is of the form $\sum_{i=1}^s a_{\bfd_i}b_{\bfd_i}$, where $s\geq 1$. Since
$$
-\nu_{p}\left(\sum_{i=1}^s a_{\bfd_i}b_{\bfd_i}\right)\leq \max_{i=1}^s\{-\nu_{p}(a_{\bfd_i})-\nu_{p}(b_{\bfd_i})\},
$$
we have
$$
\max_{\mbox{$c\in C$}}\{-\nu_{\frakp}(c)\}\leq \max_{\bfd}\{-\nu_{\frakp}(a_\bfd)\}+\max_{\bfd}\{-\nu_{\frakp}(b_\bfd)\}
$$
where $C$ is the set of all coefficients of $fg$.
It follows that
$
T(fg)\leq T(f)+T(g).
$

2.  The assertion follows from the fact that
\begin{align*}
    -\nu_\frakp(a_\bfd+b_\bfd) &\leq \max_\bfd \{-\nu_\frakp(a_\bfd),-\nu_\frakp(b_\bfd)\} \leq \max_{\bfd}\{0,-\nu_\frakp(a_\bfd),-\nu_\frakp(b_\bfd)\} \\
                    &\leq \max_{\bfd}\{0,-\nu_\frakp(a_\bfd)\}+\max_{\bfd}\{0,-\nu_\frakp(b_\bfd)\}.
\end{align*}
\end{proof}

\begin{prop}\label{prop:heightofzero}
Let $f=x^n+a_{n-1}x^{n-1}+\dots+a_0$ where $n>0$ and $a_i\in \overkt$. Suppose that $\alpha$ is a zero of $f$ in $\overkt$ and $R$ is a finite extension of $k(t)$ containing $\alpha$ and all $a_i$. Then for each place $\frakp$ of $R$,
$$
\max \{0, -\nu_{\frakp}(\alpha)\} \leq \max \{0,-\nu_\frakp(a_0),\dots,-\nu_\frakp(a_{n-1})\}.
$$
\end{prop}

\begin{proof}
The assertion is clear if $\alpha\in k$ or $\frakp$ is not a pole of $\alpha$. Assume $\alpha \in \overkt\setminus k$ and $\frakp$ is a pole of $\alpha$. Then
$$
\nu_{\frakp}(\alpha^n)=\nu_p\left(\sum_{i=0}^{n-1}a_i \alpha^i\right)\geq \min_{i=0}^{n-1} \left\{i\nu_{\frakp}(\alpha)+\nu_{\frakp}(a_i)\right\}=  i_0 \nu_\frakp(\alpha)+\nu_\frakp(a_{i_0})
$$
for some $i_0$ with $0\leq i_0 \leq n-1$.
This together with the fact that $\nu_\frakp(\alpha)<0$ implies that
$$
\nu_{\frakp}(\alpha)\geq \frac{\nu_{\frakp}(a_{i_0})}{n-i_0}\geq \nu_{\frakp}(a_{i_0})\geq \min_{i=0}^ {n-1}\left\{ \nu_\frakp(a_i)\right\},
$$
i.e.
$$
-\nu_{\frakp}(\alpha)\leq \max_{i=0}^ {n-1} \left\{-\nu_{\frakp}(a_i)\right\}.
$$
Consequently, one has that
$$
\max\{0,-\nu_{\frakp}(\alpha)\}\leq \max\{0,-\nu_{\frakp}(a_0),\dots,-\nu_{\frakp}(a_{n-1})\}.
$$
\end{proof}
\begin{corollary}
\label{cor:heightofzero}
Let $f$ be a polynomial in $\overkt[x]$ and $\alpha$ a zero of $f$ in $\overkt$. Then $T(\alpha)\leq T(f)$.
\end{corollary}
The equality in Corollary~\ref{cor:heightofzero} may hold as shown in the following example.
 \begin{example}
 Let
$$
f=x^2-(t^3+1)x+t^3=(x-t^3)(x-1),
$$
then $T(t^3)=T(f)=3$.
\end{example}

\begin{lemma}\label{lm:factor}
Assume $f,g\in \overkt[x]$ and $g$ is a factor of $f$. Then
$
T(g)\leq T(f).
$
\end{lemma}
\begin{proof}
Without loss of generality, we may assume both $f$ and $g$ are monic. Write
$$
f=x^n+\sum_{i=0}^{n-1}a_i x^i, \,\,a_i\in \overkt.
$$
We first show that if $f=gh$ and $\deg g=1$ then $T(h)\leq T(f)$. Suppose that
$$h=x^{n-1}+\sum_{j=0}^{n-2}b_jx^j,\,\,g=x+\alpha
$$
where $\alpha, b_j\in \overkt$. Let $R=k(t,\alpha,b_0,\dots,b_{n-2})$. For each place $\frakp$ of $R$, denote
$$
 N_\frakp=\max \{0,-\nu_\frakp(a_0),\dots,-\nu_p(a_{n-1})\}.
$$
Then $\max\{0,-\nu_\frakp(a_i)\}\leq N_\frakp$ for all $i=0,\dots,n-1$ and all places $\frakp$ of $R$.  From $f=gh$, one has that
$$
   b_{n-2}+\alpha=a_{n-1},\,\, b_0\alpha=a_0,\,\, b_i\alpha+b_{i-1}=a_i, \,\,i=1,\dots,n-2.
$$
We claim that $\max\{0,-\nu_\frakp(b_j)\}\leq N_\frakp$ for all $j=0,\dots,n-2$ and all places $\frakp$ of $R$. For $j=n-2$, one has that
$$
  \nu_\frakp(b_{n-2})=\nu_\frakp(a_{n-1}-\alpha)\geq \min\{\nu_\frakp(a_{n-1}),\nu_\frakp(\alpha)\}.
$$
This implies that $\max\{0,-\nu_\frakp(b_{n-2})\}\leq \max \{0,-\nu_\frakp(a_{n-1}), -\nu_\frakp(\alpha)\}$. By Proposition~\ref{prop:heightofzero}, we have that
$\max\{0,-\nu_\frakp(\alpha)\}\leq N_\frakp$. Hence $\max\{0,-\nu_\frakp(b_{n-2})\}\leq N_\frakp$ for all places $\frakp$ of $R$. Now assume that there is a place $\frakq$ of $R$ and $j_0$ with $0\leq j_0<n-2$ such that $\max\{0,-\nu_\frakq(b_{j_0})\}>N_\frakq$ but $\max\{0,-\nu_\frakq(b_{j_0+1})\}\leq N_\frakq$. Then one has that $\nu_\frakq(b_{j_0})<0$ and $\nu_\frakq(b_{j_0})<\nu_\frakq(a_i)$ for all $i=0,\dots,n-1$. On the other hand,  since $\max\{0,-\nu_\frakq(b_{j_0+1})\}\leq N_\frakq$, $-\nu_\frakq(b_{j_0})>N_\frakq\geq -\nu_\frakq(b_{j_0+1})$ i.e. $\nu_\frakq(b_{j_0+1})<\nu_{\frakq}(b_{j_0})<0$.
 From $\alpha b_{j_0+1}=a_{j_0+1}-b_{j_0}$,  one has that
 $$\nu_\frakq(\alpha b_{j_0+1})=\nu_\frakq(\alpha)+\nu_\frakq(b_{j_0+1})=\min\{\nu_\frakq(a_{j_0+1}),\nu_\frakq(b_{j_0})\}=\nu_{\frakq}(b_{j_0}).$$
 The last equality holds because $\nu_{\frakq}(b_{j_0})<\nu_\frakq(a_i)$ for all $i$. Thus $\nu_\frakq(\alpha)<0$ which implies that
 $$
    \nu_\frakq(\alpha b_{j_0})=\nu_\frakq(\alpha)+\nu_\frakq(b_{j_0})<\nu_\frakq(a_i)
 $$
 for all $i=0,\dots,n-1$.
 As $b_{j_0-1}=a_{j_0}-\alpha b_{j_0}$, one has that
$
     \nu_q(b_{j_0-1})=\nu_q(\alpha)+\nu_q(b_{j_0})<0
$
 and furthermore $\nu_q(b_{j_0-1})<\nu_q(a_i)$
 for all $i=0,\dots,n-1$.
In other words, $\max\{0,-\nu_q(b_{j_0-1})\}> N_q$. Applying a similar argument to the equalities $b_j=a_{j+1}-\alpha b_{j+1}$ for $j=j_0-2,\dots,0$ successively yields that $\nu_\frakq(b_j)<0$ and $\max\{0,-\nu_\frakq(b_j)\}>N_\frakq$ for all $j=j_0-2,\dots,0$. However, one has that $\alpha b_0=a_0$. This implies that $\nu_\frakq(a_0)<\nu_\frakq(b_0)<0$ and thus $N_\frakq\geq -\nu_\frakq(a_0)>-\nu_\frakq(b_0)>0$. That is to say, $N_\frakq\geq \max\{0,-\nu_\frakq(b_0)\}$, a contradiction. This proves the claim. This claim and Remark~\ref{rem:heights} imply that $T(h)\leq T(f)$.

Now we prove the assertion by induction on $\deg(f)$. The base case $\deg(f)=1$ is obvious. Suppose that the assertion holds for $\deg(f)\leq n$. Consider the case $\deg(f)=n+1$. If $f=g$ then there is nothing to prove. Suppose that $f\neq g$. Then there is $\beta\in \overkt$ such that $(x+\beta)g$ divides $f$. Let $f=(x+\beta)\tilde{f}$ for some $\tilde{f}\in \overkt[x]$. Then $T(\tilde{f})\leq T(f)$ and $g$ divides $\tilde{f}$. By induction hypothesis, one has that $T(g)\leq T(\tilde{f})\leq T(f)$.
\end{proof}

\begin{corollary}
\label{cor:factor}
Assume $f(x,y),g(x,y)\in \overkt[x,y]$ and $g(x,y)$ is a factor of $f(x,y)$. Then
$
T(g(x,y))\leq T(f(x,y)).
$
\end{corollary}
\begin{proof}
Suppose that $f(x,y)=\sum_{i,j} c_{i,j}x^i y^j$ where $c_{i,j}\in \overkt$. Let $d$ be an integer greater than $\tdeg(f(x,y))$. One has that
$$
    f(x,x^d)=\sum_{i,j} c_{i,j} x^{i+jd}.
$$
Note that for $0\leq i,j,l,m<d$, $i+jd=l+md$ if and only if $(i,j)=(l,m)$. This implies that the set of the coefficients of $f(x,y)$ coincides with that of $f(x,x^d)$. Hence $T(f(x,y))=T(f(x,x^d))$. Similarly, $T(g(x,y))=T(g(x,x^d))$. It is clear that $g(x,x^d)$ is still a factor of $f(x,x^d)$. By Lemma~\ref{lm:factor}, $T(g(x,x^d))\leq T(f(x,x^d))$. Thus $T(g(x,y))\leq T(f(x,y))$.
\end{proof}

\begin{prop}\label{prop:resultant}
Assume $f,g\in \overkt[\bfy,z]$, then
$$
T({\rm res}_{z}(f,g))\leq \deg(g,z)T(f)+\deg(f,z)T(g)
$$
where ${\rm res}_{z}(f,g)$ is the resultant of $f$ and $g$ with respect to $z$.
\end{prop}
\begin{proof}
The assertion is clear that if ${\rm res}_{z}(f,g)=0$. Consider the case ${\rm res}_{z}(f,g)\neq 0$.
Assume $\deg(f,z)=n,\deg(g,z)=m$. Write
$$
f=\sum_{i=0}^n a_i(\bfy)z^i,\quad g=\sum_{i=0}^m b_i(\bfy)z^i
$$
where $a_i(\bfy),b_i(\bfy)\in \overkt[\bfy]$.  Denote further by $C_1, C_2$ the sets of the coefficients in $\bfy,z$ of $f, g$  respectively.
Then
$$
{\rm res}_z(f,g)=
\begin{vmatrix}
a_n   & a_{n-1}  &\cdots  & a_0   \\
         & \ddots    &  \ddots       &   & \ddots   \\
         &           &  a_n    & a_{n-1}   &\cdots  & a_0 \\
b_m   & b_{m-1}    &\cdots & b_0   \\
         & \ddots    &    \ddots      &  & \ddots     \\
         &           &  b_m    & b_{m-1}   &\cdots  & b_0 \\
\end{vmatrix}.
$$
By the definitions of determinant, we can write
$$
{\rm res}_z(f,g)=\sum_{\bfd} \left(\sum_{j=1}^{\ell_\bfd}\beta_{\bfd,j} \bfm_{\bfd,j}\bfn_{\bfd,j} \right)\bfy^{\bfd}
$$
where $\beta_{\bfd,j}, \ell_\bfd\in \Z, \ell_\bfd\geq 0$, $\bfm_{\bfd,j}$ is a monomial in $C_1$ with total degree $m$ and $\bfn_{\bfd,j}$ is a monomial in $C_2$ with total degree $n$. Let $R=k(t,C_1,C_2)$.
 For each place $\frakp$ of $R$, we have
\begin{align*}
-\nu_{\frakp}\left(\sum_j \beta_{\bfd,j} \bfm_{\bfd,j}\bfn_{\bfd,j}\right)
&\leq \max_j \{-\nu_{\frakp}( \bfm_{\bfd,j}\bfn_{\bfd,j})\}\\
&\leq m\max_{c\in C_1}\{-\nu_{\frakp}(c)\}+n\max_{c\in C_2}\{-\nu_{\frakp}(c)\}.
\end{align*}
Therefore by Remark~\ref{rem:heights},
$$
T({\rm res}_{z}(f,g))\leq mT(f)+nT(g).
$$
\end{proof}

\section{Degrees and Heights on Riemann-Roth spaces}
Throughout this section, $\calR$ denotes an algebraic function field of one variable over $\overkt$.
Let $\bfxi=(\xi_0,\xi_1,\xi_2)\in \calR^3$ be such that $\calC(\bfxi)$ a plane projective model of $\calR$, i.e. $\calR=\overkt(\xi_0/\xi_2,\xi_1/\xi_2)$. Each $h\in \calR$ can be presented by $G(\bfxi)/H(\bfxi)$ where $G, H$ are two homogeneous polynomials in $\overkt[x_0,x_1,x_2]$ of the same degree and having no common factors, and $H(\bfxi)\neq 0$. We call $(G,H)$ a representation of $h$. This section shall focus on determining the degrees and heights of representations of elements in Riemann-Roch spaces.
 There are several algorithms for computing the bases of Riemann-Roth spaces (see for example \cite{hess,huang-ierardi}). However no existing algorithm provided explicit bounds for the degrees and heights of $G$ and $H$ where $(G,H)$ represents an element in these bases. These bounds play an essential role in estimating the heights of points on a plane algebraic curve. In this section, we shall follow the algorithm developed in \cite{huang-ierardi}  to obtain these bounds.  For this purpose, we need to resolve singularities of a given plane algebraic curve to  obtain the one with only ordinary singularities. This can be done by a sequence of quadratic transformations.

In this section, unless otherwise stated,  $F$ always stands for an irreducible homogeneous polynomial in $\overkt[x_0,x_1,x_2]$ of degree $n>0$ which defines a plane projective model of $\calR$, i.e. there is $\bfxi=(\xi_0,\xi_1,\xi_2)\in \calR^3$ such that $\calR=\overkt(\xi_0/\xi_2,\xi_1/\xi_2)$ and $F(\bfxi)=0$.

\subsection{Quadratic transformations}\label{subsec:quadratic}
Let $D$ be a divisor in $\calR$. Due to Proposition on page of \cite{fulton}, there is a birational transformation $\calB$ such that the transformation of $F=0$ under $\calB$ is a plane projective curve with only ordinary singularties, moreover $\calB$ can be chosen to be the composition of a sequence of suitable quadratic transformations. In this subsection, we shall investigate the degree and height of the transformation of $F=0$ under a quadratic transformation.
\begin{define}
\begin{enumerate}
\item
$\calL$ stands for a projective change of coordinates on $\bP^2(\overkt)$ that is defined as
$
  \calL(\bfc)=\bfc M_\calL
$
where $M_\calL\in \GL_3(\overkt)$, and $\calQ$  denotes  the standard quadratic transformation that  is defined as
$$
\calQ(\bfc)=(c_1c_2,c_0c_2,c_0c_1)
$$
where $\bfc=(c_0,c_1,c_2)$.
\item The height of $\calL$, denoted by $T(\calL)$, is defined as $T(M_\calL)$.
\end{enumerate}
\end{define}
\begin{notation}
\begin{enumerate}
\item
$F^\calL$ stands for $F((x_0,x_1,x_2)M_\calL)$.
\item
$F^\calQ$ stands for the irreducible polynomial  $\tilde{F}$ satisfying
$$
     F(x_1x_2,x_0x_2,x_0x_1)=x_0^{d_0}x_1^{d_1}x_2^{d_2}\tilde{F}
$$
where $d_i\geq 0$.
\end{enumerate}
One sees that $\bV(F^\calL)$ (resp. $\bV(F^\calQ)$) is the variaty $\calL^{-1}(\bV(F))$ (resp. the projective closure of $\calQ^{-1}(\bV(F)\setminus \bV(x_0x_1x_2)$).
\end{notation}
\begin{remark}
\label{rem:quadratictransformation}
$\calQ$ is bijective on $\bP^2(\overkt)\setminus \bV(x_0x_1x_2)$ and $\calQ^{-1}=\calQ$.
\end{remark}
Let us first bound the heights of the common points of two algebraic curves in $\bP^2(\overkt)$.
\begin{prop}\label{prop:intersection}
Let $F,G$ be two homogenenous polynomials in $\overkt[x_0,x_1,x_2]$ of degree $n,m$ respectively. Suppose that $F$ and $G$ have no common factor, and $\bfc\in \bP^2(\overkt)$ is a common point of $F=0$ and $G=0$. Then
$$
T(\bfa)\leq 2(mT(F)+nT(G)).
$$
Furthermore, if $G=c_0x_1+c_1x_1+c_2x_2$ with $c_i\in k$ then $T(\bfa)\leq T(F)$.
\end{prop}
\begin{proof}
Let
$H_i(x_j,x_l)={\rm res}_{x_i}(F,G)$ where $ \{i,j,l\}=\{0,1,2\}$.
Proposition \ref{prop:resultant} implies that
$
T(H_i)\leq mT(F)+nT(G).
$
Without loss of generality, suppose $\bfa=(1,a_1,a_2)$. Since $\bfa$ is a common point of $F=0$ and $G=0$,
$$
H_2(1,a_1)=H_1(1,a_2)=0.
$$
It follows from Corollary \ref{cor:heightofzero} that
$
T(a_i)\leq mT(F)+nT(G)
$ for all $i=1,2$.
Let $R=k(t,a_1,a_2)$ and $\frakp$ a place of $R$.  Then
$$
\max\{0,-\nu_{\frakp}(a_1),-\nu_{\frakp}(a_2)\}\leq
\max\{0,-\nu_{\frakp}(a_1)\}+\max\{0,-\nu_{\frakp}(a_2)\}.
$$
Whence
$$
T(\bfa)\leq T(a_1)+T(a_2)\leq 2(mT(F)+nT(G)).
$$
It remains to pove the second assertion. Since $a_0=1\neq 0$, not all $c_1,c_2$ are zero. Without loss of generality, assume that $c_1\neq 0$.  Substituting $a_1=-(c_0+a_2c_2)/c_1$ into $F=0$ yields that
$$F(1,-(c_0+a_2c_2)/c_1,a_2)=0.$$
This implies that $T(a_2)\leq T(F)$. On the other hand, one sees that
$$
    \nu_\frakp(a_1)=\nu_\frakp(-(c_0+a_2c_2)/c_1)\geq \min\{\nu_\frakp(c_0),\nu_\frakp(a_2)\}.
$$
So
$$
   \max\{0, -\nu_\frakp(a_1),-\nu_\frakp(a_2)\}\leq \max\{0, -\nu_\frakp(a_2)\}.
$$
which results in $T(\bfa)\leq T(a_2)\leq T(F)$.
\end{proof}
\begin{corollary}\label{cor:singularity}
If $\bfa\in \bP^2(\overkt)$ is a singular point of $F=0$, then
$
T(\bfa)\leq 2(2n-1)T(F).
$
\end{corollary}

\begin{lemma}
\label{lm:lineartransformation}
Suppose that $\calL$ is a projective change of coordinates. Then
\begin{enumerate}
\item
$T(F^\calL)\leq T(F)+\deg(F)T(\calL)$;
\item
  for each $\bfc\in \bP^2(\overkt)$, $T(\calL(\bfc))\leq T(\bfc)+T(\calL)$.
\end{enumerate}
\end{lemma}
\begin{proof}
Suppose that $M_\calL=(a_{i,j})$ and
$$F=\sum_{i=0}^n \sum_{j=0}^{n-i} c_{i,j}x_0^ix_1^j x_2^{n-i-j}$$
where $n=\deg(F)$ and $c_{i,j}\in \overkt$.

1. One has that
$$
  F^\calL=\sum_{i=0}^n \sum_{j=0}^{n-i} c_{i,j} \left(\sum_{l=1}^3 a_{l,1}x_{l-1}\right)^i  \left(\sum_{l=1}^3 a_{l,2}x_{l-1}\right)^j  \left(\sum_{l=1}^3 a_{l,3}x_{l-1}\right)^{n-i-j}.
$$
Let $\rho$ be a coefficient of $F^\calL$ viewed as polynomial in $x_0,x_1,x_2$. Then $\rho$ is a $k$-linear combination of monomials
$ c_{i,j}a_{1,1}^{e_{1,1}} \dots a_{3,3}^{e_{3,3}}$ with $\sum e_{i,j}=n$. Let $R$ be a finite extension of $k(t)$ containing all $c_{i,j}$ and $a_{i,j}$. Suppose that $\frakp$ is a place of $R$. Then one has that
\begin{align*}
 \nu_{\frakp}(\rho) &\geq \min_{i,j, i',j'}\left\{\nu_{\frakp}(c_{i,j})+\sum_{i',j'} e_{i',j'}\nu_\frakp(a_{i',j'})\right\}\\
 &\geq \min_{i,j} \{\nu_{\frakp}(c_{i,j})\}+\min_{i,j} \left\{\sum_{i,j} e_{i,j}\nu_\frakp(a_{i,j})\right\},
\end{align*}
i.e.
\begin{align*}
   -\nu_{\frakp}(\rho)& \leq \max_{i,j}\{-\nu_{\frakp}(c_{i,j})\}+\max_{i,j} \left\{-\sum_{i,j} e_{i,j}\nu_\frakp(a_{i,j})\right\}\\
   & \leq \max_{i,j}\{-\nu_\frakp(c_{i,j})\}+n\max_{i,j}\{-\nu_\frakp(a_{i,j})\}.
\end{align*}
Therefore $T(F^\calL)\leq T(F)+nT(\calL)$ due to Remark~\ref{rem:heights}.

2. Suppose that $\bfc=(c_0,c_1,c_2)$ and $\calL(\bfc)=(b_0,b_1,b_2)$.  Then $b_i=\sum_{j=1}^3 a_{j,i}c_{j-1}$. Let $R$ be a finite extension of $k(t)$ containing all $c_i$ and $a_{i,j}$, and $\frakp$ a place of $R$. Then
$$
   \nu_\frakp(b_i)=\nu_\frakp\left(\sum_{j=1}^3 a_{j,i}c_{j-1}\right)\geq \min_j \{\nu_\frakp(a_{j,i})+\nu_\frakp(c_{j-1})\}
$$
i.e.
$$
    -\nu_\frakp(b_i)\leq \max_j \{-\nu_\frakp(a_{j,i})-\nu_\frakp(c_{j-1})\}\leq \max_j \{-\nu_\frakp(a_{j,i})\}+\max_j \{-\nu_\frakp(c_{j-1})\}.
$$
So $T(\calL(\bfc))\leq T(\bfc)+T(\calL)$.
\end{proof}

\begin{corollary}
\label{cor:lineartransformation}
Suppose that $\bfc=(c_0,c_1,1)\in \bP^2(\overkt)$. Let $\calL$ be a projective change of coordinates with
\begin{equation}
\label{eq:lineartransformation}
   M_\calL=\begin{pmatrix}
   a_1 & a_2 & a_3 \\
   a_4 & a_5 & a_6 \\
   c_0  & c_1 & 1
   \end{pmatrix}
\end{equation}
where $a_i\in k$. Then
\begin{enumerate}
\item $\calL((0,0,1))=\bfc$;
\item $T(F^{\calL}), T(F^{\calL^{-1}})\leq T(F)+\deg(F)T(\bfc)$;
\item for each $\bfb\in \bP^2(\overkt)$, $T(\calL(\bfb)), T(\calL^{-1}(\bfb))\leq T(\bfb)+T(\bfc)$.
\end{enumerate}
\end{corollary}
\begin{proof}
The first assertion is obvious. The second and third assertions follows from Lemma~\ref{lm:lineartransformation} and the fact that $T(M_\calL)=T(\bfc)$ and $T(M_{\calL^{-1}})\leq T(\bfc)$.
\end{proof}

\begin{define}
\label{def:excellentposition}
\begin{enumerate}
\item The points $(1,0,0), (0,1,0), (0,0,1)\in \bP^2(\overkt)$ are called fundamental points.
\item Assume $(0,0,1)$ is a singular point of $F=0$ with multiplicity $r$. $F=0$ is said to be in excellent position if  it satisfies the following two conditions:
\begin{enumerate}
\item the line $x_2=0$ intersects $F=0$ in $n$ distinct non-fundamental points;
\item the lines $x_0=0, x_1=0$ intersect $F=0$ in $n-r$ distinct points other than fundamental points.
\end{enumerate}
\end{enumerate}
\end{define}

\begin{lemma}
\label{lm:lineartransformation2}
Suppose that $\bfc=(c_0,c_1,c_2)$ is a singular point of $F=0$ of multiplicity $r$. There is a projective change of coordinates $\calL$ with $M_\calL$ having the form (\ref{eq:lineartransformation}) such that
$\calL((0,0,1))=\bfc$ and $F^\calL=0$ is in excellent position.
\end{lemma}
\begin{proof}
 Denote $\bfy=(y_1,\dots,y_6)$. Let $\calL'_\bfy$ be the projective change of coordinates with $M_{\calL'_\bfy}$ of the form
\[
        \begin{pmatrix}
              y_1 & y_2 & y_3\\
              y_4 & y_5 & y_6 \\
              c_0 & c_1 &  c_2
        \end{pmatrix}.
 \]
One sees that there are polynomials $f_1,\dots,f_s\in \overkt[y_1,\dots,y_6]$ such that  $\calL'_\bfb$ with $\bfb \in \overkt^6$ satisfies the required conditions if and only if
$$
     \bfb \in S=\left\{ \bfb \in \overkt^6 \mid \forall\,i=1,\dots,s, \,f_i(\bfb)\neq 0\right\}.
$$
Note that if $\calL$ is a projective change of coordinates such that $\calL(\bfc)=(0,0,1)$ then $\calL=\calL'_\bfb$ for some $\bfb\in \overkt^6$.
 Due to Lemma 1 on page of \cite{fulton}, there are projective changes of coordinates satisfying the above conditions. In other words, $S\neq \emptyset$.  Therefore $S\cap k^6 \neq \emptyset$. For  every $\bfb \in S\cap k^6$, $\calL'_\bfb$ is as required.
\end{proof}

\begin{lemma}
\label{lm:standardtransformation}
\begin{enumerate}
\item
    $T(F^{\calQ})=T(F)$;
\item
    For each $\bfa=(a_0,a_1,a_2)\in \bP^2(\overkt)$, $T(\calQ(\bfa))\leq 2T(\bfa)$.
\end{enumerate}
\end{lemma}
\begin{proof}
$1.$ Assume that $F=\sum_{i=0}^n \sum_{j=0}^{n-i} c_{i,j}x_0^i x_1^j x_2^{n-i-j}$ where $c_{i,j}\in \overkt$. Then
\begin{align*}
   F(x_1x_2,x_0x_2,x_0x_1)&=\sum_{i=0}^n \sum_{j=0}^{n-i} c_{i,j} (x_1x_2)^i (x_0x_2)^j (x_0x_1)^{n-i-j}\\
   &=\sum_{i=0}^n \sum_{j=0}^{n-i} c_{i,j}  x_0^{n-i}x_1^{n-j}x_2^{i+j}.
\end{align*}
From this, one sees that the set of coefficients of $F$ is equal to that of $F^{\calQ}$. Hence $T(F)=T(F^{\calQ})$.

$2.$ One has that $\calQ(\bfa)=(a_1a_2,a_0a_2,a_0a_1)$. Let $R$ be a finite exntesion of $k(t)$ containing all $a_i$ and $\frakp$ a place of $R$. Note that
$$
     \nu_{\frakp}(a_ia_j)=\nu_{\frakp}(a_i)+\nu_{\frakp}(a_j) \geq 2\min\{\nu_{\frakp}(a_0),\nu_{\frakp}(a_1),\nu_{\frakp}(a_2)\},
$$
i.e. $-\nu_{\frakp}(a_ia_j)\leq 2\max\{-\nu_{\frakp}(a_0),-\nu_{\frakp}(a_1),-\nu_{\frakp}(a_2)\}$. So $T(\calQ(\bfa))\leq 2T(\bfa)$.
\end{proof}

 \begin{define}
\begin{enumerate}
\item
We call a projective change of coordinates in Lemma~\ref{lm:lineartransformation2} a projective change of coordinates centered at $\bfc$.
\item
\label{def:qtransformation} Let $\calL_\bfc$ be a projective change of coordinates centered at $\bfc$ and $\calQ$ the stand quadratic transformation. We call $\calQ_\bfc=\calL_\bfc\circ \calQ$, the composition of $\calQ$ and $\calL_\bfc$, a quadratic transformation centered at $\bfc$.
\end{enumerate}
\end{define}
\begin{notation}
Let $\calQ_\bfc$ be a quadratic transformation centered at $\bfc$. We shall denote $F^{\calQ_\bfc}=(F^{\calL_\bfc})^{\calQ}$.
\end{notation}
\begin{corollary}
\label{cor:qtransformation}
Let $\bfc$ be a singular point of $F=0$ and $\calQ_\bfc$ a quadratic transformation centered at $\bfc$. Then
\begin{enumerate}
\item
    $T(F^{\calQ_\bfc})\leq T(F)+\deg(F)T(\bfc)$;
\item
    for $\bfa \in \bP^2(\overkt)$, $T(\calQ_\bfc^{-1}(\bfa))\leq 2(T(\bfc)+T(\bfa))$;
\end{enumerate}
\end{corollary}
\begin{proof}
$1$ and $2$ follow from the fact that $\calQ^{-1}=\calQ$ and Lemmas~\ref{lm:lineartransformation} and~\ref{lm:standardtransformation}.
\end{proof}
\begin{prop}
\label{prop:centers}
Let $\calC(\bfxi)$ be a plane projective model of $\calR$ defined by $F$ and $\frakP$ a place of $\calR$. Let $\bfa$ be the center of $\frakP$ with respect to $\bfxi$. Assume that $\calQ_\bfc$ is a quadratic transformation centered at $\bfc$ for some singular point $\bfc$ of $F=0$ and $\bfa'$ is the center of $\frakP$ with respect to $\calQ_\bfc^{-1}(\bfxi)$.
Then
$$T(\bfa')\leq \max\{2(T(\bfc)+T(\bfa)), T(F)+\deg(F)T(\bfc)\}.$$
\end{prop}
\begin{proof}
We first claim that if $\bfa \neq \bfc$ then $\calQ_\bfc^{-1}(\bfa)\neq (0,0,0)$. Otherwise assume that $\calQ_\bfc^{-1}(\bfa)=\calQ^{-1}\calL_\bfc^{-1}(\bfa)=(0,0,0)$. Then $\calL_\bfc^{-1}(\bfa)$ is a fundamental point of $F^{\calL_\bfc}=0$. Since neither $(1,0,0)$ nor $(0,1,0)$ is a point of $F^{\calL_\bfc}=0$. One has that $\calL_\bfc^{-1}(\bfa)=(0,0,1)$. Hence $\bfa=\bfc$. This proves our claim.

Suppose that $\bfa \neq \bfc$. Then $\calQ_\bfc^{-1}(\bfa)\neq (0,0,0)$ and thus $\bfa'=\calQ_\bfc^{-1}(\bfa)$ by Proposition~\ref{prop:centertransformation}. Corollary~\ref{cor:qtransformation} then implies that  $T(\bfa')\leq 2(T(\bfa)+T(\bfc)).$
Now suppose that $\bfa=\bfc$. Denote $\bfxi'=\calL_\bfc^{-1}(\bfxi)=(\xi_0',\xi_1',\xi_2')$. By Proposition~\ref{prop:centertransformation} again, $(0,0,1)$ is the center of $\frakP$ with respect to $\bfxi'$. Suppose that $u$ is a local uniformizer of $\frakP$ and $\ell_i=\nu_{\frakP}(\xi_i')$. From the definition of center, one sees that $\ell_i>\ell_2$ for all $i=0,1$.
 Write $\xi_i'=u^{\ell_i}(c_i+u \eta_i)$
 where $c_i\in \overkt\setminus\{0\}$, $\eta_i\in \calR$ with $\nu_{\frakP}(\eta_i)\geq 0$.
 One then  has that
 \begin{align*}
 \calQ^{-1}(\bfxi')&=(\xi_1'\xi_2',\xi_0'\xi_2',\xi_0'\xi_1')\\
 &=\left(u^{\ell_1+\ell_2}(c_1c_2+u\tilde{\eta_0}), u^{\ell_0+\ell_2}(c_0c_2+u\tilde{\eta_1}), u^{\ell_1+\ell_0}(c_0c_1+u\tilde{\eta_2})\right)
 \end{align*}
 where $\tilde{\eta_i}\in \calR$ with $\nu_{\frakP}(\tilde{\eta_i})\geq 0$. Set
 $$
   \mu= \min \{\nu_{\frakP}(\xi_1'\xi_2'), \nu_{\frakP}(\xi_0'\xi_2'),\nu_{\frakP}(\xi_1'\xi_0')\}.
 $$
 Since both $\ell_0$ and $\ell_1$ are greater than $\ell_2$, $\mu=\ell_0+\ell_2$ or $\ell_1+\ell_2$.
 In the case that $\mu=\ell_0+\ell_2=\ell_1+\ell_2$, one has that $\bfa'=(c_1c_2,c_0c_2,0)=c_2(c_1,c_0,0)$. So $\bfa' \in \bV(F^{\calQ_\bfc})\cap \bV(x_2)$. By Propositon~\ref{prop:intersection}, $T(\bfa')\leq T(F^{\calQ_\bfc})\leq T(F)+\deg(F)T(\bfc)$. In other cases, one sees that $\bfa'$ is a fundamental point and so $T(\bfa')=0$.
Hence, in each case, one has that
 $$
    T(\bfa')\leq \max\{2(T(\bfc)+T(\bfa)), T(F)+\deg(F)T(\bfc)\}.
 $$
\end{proof}

\subsection{Degrees and Heights for Riemann-Roch Spaces}\label{subsec:ordinary}
Let $D=\sum_{i=1}^m n_i \frakP_i$ be a disivor in $\calR$ where $n_i\neq 0$. Let $\calC(\bfxi)$ defined by $F$ be a plane projective model of $\calR$. Suppose that $h\in \frakL(D)$ and $h=G(\bfxi)/H(\bfxi)$ where $G,H$ are two homogeneous polynomials of the same degree in $\overkt[x_0,x_1,x_2]$. In this subsection, we shall estimate $\deg(G)$ and $T(G), T(H)$ in terms of $\deg(F)$ and $T(F)$. For this, we introduce the following notations and definitions.
\begin{notation}
Let  $\calC(\bfxi)$ be a plane projective model of $\calR$ and $D$ a divisor.
\begin{enumerate}
\item
    $\calS_\bfxi(D):=\{\mbox{the centers of places in $\supp(D)$ with respect to $\bfxi$}\}$.
\item
  $T_\bfxi(D):=\max\left\{T(\bfc) \mid \bfc\in \calS_\bfxi(D)\right\}$.
\end{enumerate}
\end{notation}
\begin{define}
\label{def:intersection}
Let $G,H$ be two homogeneous polynomials in $\overkt[x_0,x_1,x_2]$ of the same degree. Write $\bfxi=(\xi_0,\xi_1,\xi_2)$.
\begin{itemize}
\item [$(1).$]
Define
$$
     \ord_{\frakP}(G(\bfxi))=\nu_{\frakP}\left(G(\bfxi)\right)-\deg(G)\min_{i=0}^2\{\nu_{\frakP}(\xi_i)\}.
$$
\item [$(2).$]
Define
$$
     \divs_\bfxi(G)=\sum \ord_{\frakP}(G(\bfxi))\frakP
$$
where the sum ranges over all places of $\calR$. Furthermore, define $$\divs_\bfxi(G/H)=\divs_\bfxi(G)-\divs_\bfxi(H).$$
\end{itemize}
\end{define}
It is easy to see that $\ord_{\frakP}(G(\bfxi))\geq 0$ and $\ord_\frakP(G(\bfxi))>0$ if and only if $G(\bfc)=0$ where $\bfc$ is the center of $\frakP$ with respect to $\bfxi$. Futhermore $\ord_{\frakP}(G(\lambda \bfxi))=\ord_{\frakP}(G(\bfxi))$ for all nonzero $\lambda \in \calR$.
\begin{remark}
\label{rm:intersectioncycle}
On page 182 of \cite{fulton}, $\ord_{\frakP}(G)$ is defined to be the order at $\frakP$ of the image of $G$  in the valutaion ring of  $\frakP$. Remark that
$$\ord_\frakP(G(\bfxi))=\nu_{\frakP}(G(\bfxi)/\xi_{i_0}^d)$$
where $\xi_{i_0}$ satisfies that $\nu_{\frakP}(\xi_{i_0})=\min_{i=0}^{2}\{\nu_\frakP(\xi_i)\}$. Under the map sending $x_j$ to $\xi_j/\xi_{i_0}$ for all $j=0,1,2$, $G$ is sent to $G(\bfxi)/\xi_{i_0}^d$ which lies in the valuation ring of $\frakP$. Therefore $\ord_\frakP$ given in Definition~\ref{def:intersection} concides with the one given in \cite{fulton} and $\divs_\bfxi(G)$ is nothing else but the intersection cycle of $G$ and $H$ (see page 119 of \cite{fulton}).
\end{remark}
The lemma below follows easily from the definition.
\begin{lemma}
\label{lm:cycle}
Suppose that $G,H\in \overkt[x_0,x_1,x_2]$ are two homogeneous polynomials  of the same degree. Then
\begin{enumerate}
\item
    $\divs_\bfxi\left(\frac{G}{H}\right)=\divs_\bfxi(G)-\divs_\bfxi(H)=\divs\left(\frac{G(\bfxi)}{H(\bfxi)}\right)$; and
\item
    $\deg(\divs_\bfxi(G))=\deg(G)\deg(F)$.
\end{enumerate}
\end{lemma}
\begin{lemma}
\label{lm:generallines}
Suppose that $\bfc=(c_0,c_1,1)$ is an ordinary singular point of $F=0$ of multiplicity $r$, $n=\deg(F)$ and $S$ is a finite set of points of $F=0$.
\begin{enumerate}
\item  Let $L_\lambda=x_0-c_0x_2-\lambda(x_1-c_1x_2)$. Then for all but a finite number of $\lambda$, $L_\lambda=0$ intersects $F=0$ in $n-r$ distinct points other than the points in $\{\bfc\}\cup  S$.
\item Write
$$
    F=F_r(x_0-c_0x_2, x_1-c_1x_2)x_2^{n-r}+\dots+ F_n(x_0-c_0x_2,x_1-c_1x_2)
$$
where $F_i(y_0,y_1)$ is a homogeneous polynomial in $y_0,y_1$ of degree $i$. Assume that $F_r(1,0)F_r(0,1)F_n(1,0)F_n(0,1)\neq 0$. Let
$$G_\lambda=\alpha (x_1-c_1x_2)x_2-(x_0-c_0x_2)x_2-\lambda (x_0-c_0x_2)(x_1-c_1x_2)$$
 where $\alpha\in \overkt\setminus\{0\}$ satisfies that $F_r(\alpha,1)=0$. Then for all but a finite number of $\lambda$, $G_\lambda=0$ intersects $F=0$ in $2n-r-1$ distinct points other than the points in $\{\bfc\}\cup S$.
\end{enumerate}
\end{lemma}
\begin{proof}
$1.$ Under the projective change of coordinates $\calL$ with $\calL(x_0)=x_0+c_0x_2, \calL(x_1)=x_1+c_1x_2$ and $\calL(x_2)=x_2$, we may assume that $\bfc=(0,0,1)$. Set $L_\lambda=x_0-\lambda x_1$ where $\lambda\in \overkt$. Substituting $x_0=\lambda x_1$ into $F$ yields that
$$
     x_1^r\left(F_r(\lambda, 1)x_2^{n-r}+F_{r+1}(\lambda, 1)x_1x_2^{n-r-1}+\dots+F_n(\lambda, 1)x_1^{n-r}\right).
$$
Set $H_\lambda(t)=F_\lambda(\lambda,1)t^{n-r}+\dots+F_n(\lambda,1)$. For every root $\gamma$ of $H_\lambda(t)=0$, one sees that $(\lambda, 1,\gamma)$ is a common point of $L_\lambda=0$ and $F=0$ other than $\bfc$. Moreover if $H_\lambda(t)=0$ has $n-r$ distinct roots then $L_\lambda=0$ intersects $F=0$ in $n-r$ distinct points other than $\bfc$. So it suffices to prove that for all but a finite number of $\lambda$, $H_\lambda(t)=0$ has $n-r$ distinct roots. Note tht substituting $x_0=\lambda x_1$ into $\partial F/\partial x_2$ yields that
$$
      x_1^r \sum_{i=r}^n (n-i) F_i(\lambda,1) x_2^{n-i-1} x_1^{i-r}.
$$
From this, one sees that if $\gamma$ is a common root of $H_\lambda(t)=0$ and $\partial H_\lambda(t)/\partial t=0$ then $(\lambda,1,\gamma)$ is a common point of $F=0$ and $\partial F/\partial x_2=0$. Since $F=0$ and $\partial F/\partial x_2=0$ have only finitely many common points, there are only finitely many $\lambda$ such that $H_\lambda(t)=0$ has multiple roots. In other words, there is only a finite number of $\lambda$ such that $L_\lambda=0$ intersects $F=0$ in less than $n-r$ distinct points other than $\bfc$. It remains to show that there are only finitely many $\lambda$ such that $(S\setminus \{\bfc\})\cap \bV(L_\lambda)\neq \emptyset$.
 Assume that $(a_0,a_1,a_2)\in S\setminus\{\bfc\}$ which lies in the line $L_\lambda=0$. If $a_1-c_1a_2\neq 0$ then
 $$\lambda=(a_0-c_2a_2)/(a_1-c_1a_2).$$  If $a_1-c_1a_2=0$ then $a_0-c_0a_2=0$ and thus $(a_0,a_1,a_2)=a_2(c_0,c_1,1)$. In other words, $(a_0,a_1,a_2)=(0,0,0)$ or $\bfc$, which is impossible.

$2.$ Similarly, we may assume that $\bfc=(0,0,1)$. Then
\begin{align*}
   G_\lambda & =\alpha x_1x_2-x_0x_2 -\lambda x_0x_1,\\
   F&=F_r(x_0,x_1)x_2^{n-r}+\dots+F_n(x_0,x_1).
\end{align*}
Applying the standard quadratic transformation $\calQ$ to $G_\lambda$ and $F$, one obtains that
\begin{align*}
G_\lambda^{\calQ}&=\alpha x_0-x_1-\lambda x_2,\\
F(x_1x_2,x_0x_2,x_0x_1)&=x_2^r\left(F_r(x_1,x_0)(x_0x_1)^{n-r}+\dots+F_n(x_1,x_0)x_2^{n-r}\right).
\end{align*}
Since $\bfc$ is an ordinary singular point and $F_r(1,0)F_r(0,1)\neq 0$, by on page of \cite{fulton} $(1,\alpha,0)$ is a simple point of $F^{\calQ}=0$. Moveover, as $F_n(1,0)F_n(0,1)\neq 0$, neither $(1,0,0)$ nor $(0,1,0)$ is a point of $F^{\calQ}=0$ and so $\deg(F^{\calQ})=2n-r$. Thus
$$
     F^\calQ=F_r(x_1,x_0)(x_0x_1)^{n-r}+\dots+F_n(x_1,x_0)x_2^{n-r}.
$$
For every common point $(\gamma_0,\gamma_1,\gamma_2)$ of $G_\lambda^{\calQ}=0$ and $F^{\calQ}=0$ with $\lambda\gamma_2\neq 0$,
$(\gamma_1\gamma_2,\gamma_0\gamma_2,\gamma_0\gamma_1)\neq (0,0,0)$ and then it is a common point of $G_\lambda=0$ and $F=0$ other than $\bfc$. Therefore it suffices to show that for all but a finite number of $\lambda$, $G_\lambda^{\calQ}=0$ and $F^{\calQ}=0$ have $2n-r-1$ distinct common points $(\gamma_0,\gamma_1,\gamma_2)$ with $\gamma_2\neq 0$. Let $\calL$ be the projective change of coordinates such that $\calL(x_0)=(x_0+x_1)/\alpha,\calL(x_1)=x_1,\calL(x_2)=x_2$. Then
$(G_\lambda^{\calQ})^\calL=x_0-\lambda x_2$. Note that $\calL^{-1}((1,\alpha,0))=(0,\alpha,0)$ which is a simple point of $(F^{\calQ})^\calL=0$. Thus
$$
     (F^{\calQ})^\calL=\tilde{F}_1(x_0,x_2)x_1^{2n-r-1}+\tilde{F}_2(x_0,x_2)x_1^{2n-r-2}+\dots+\tilde{F}_{2n-r}(x_0,x_2).
$$
By $(1)$, for all but a finite number of $\lambda$, $(G_\lambda^{\calQ})^\calL=0$ intersects $(F^{\calQ})^\calL=0$ in $2n-r-1$ distinct points $(\gamma_0',\gamma_1',\gamma_2')$ with $\gamma_2'\neq 0$. Remark that if $(\gamma'_0,\gamma_1',\gamma_2')$ is a common point of $(G_\lambda^{\calQ})^\calL=0$ and $(F^{\calQ})^\calL=0$ with $\gamma_2'\neq 0$ then $(\gamma_0'+\beta/\alpha \gamma_1',\gamma_1',\gamma_2')$ is a common point of $G_\lambda^{\calQ}=0$ and $F^{\calQ}=0$ with $\gamma_2'\neq 0$. These imply that for all but a finite number of $\lambda$, $G_\lambda^{\calQ}=0$ intersects $F^{\calQ}=0$ in $2n-r-1$ distinct points $(\gamma_0,\gamma_1,\gamma_2)$ with $\gamma_2\neq 0$.

Finally, we need to prove that there are only finitely many $\lambda$ such that $S\setminus \{\bfc\}\cap \bV(G_\lambda)\neq \emptyset$. Assume that $\bfa=(a_0,a_1,a_2)\in S\setminus\{\bfc\}$ which lies in $G_\lambda=0$. We claim that $(a_0-c_0a_2)(a_1-c_1a_2)\neq 0$.  Suppose on the contrary that $(a_0-c_0a_2)(a_1-c_1a_2)=0$. Then by $G_\lambda(\bfa)=0$, one sees that either $a_2=0$ or both $a_0-c_0a_2$ and $a_1-c_1a_2$ are zero. This implies that $\bfa$ must be one of three points $(1,0,0), (0,1,0), a_2(c_0,c_1,1)$.  This is impossible and then our claim holds. It follows from the claim that $\lambda$ is uniquely determined by $\bfa$.
\end{proof}

\begin{prop}\label{prop:simplification}
Suppose that $F=0$ has only ordinary singularities and $D$ is an effective divisor in $\calR$. Let $D'$ be a divisor in $\calR$.
\begin{enumerate}
\item
Assume further that $D=\sum_{i=1}^r \frakP_i$ where  all $\frakP_i$ have the same center which is a point of $F=0$ with multiplicity $r$.
Then there is a linear homogeneous polynomial $G$ in $\overkt[x_0,x_1,x_2]$ such that
$$
\divs_\bfxi(G)=D+A
$$
where $A$ is a very simple and effective divisor of degree $n-r$, $\supp(A)\cap (\supp(D')\cup\{\frakP_1,\dots,\frakP_r\})=\emptyset$, and
$$
T(G)\leq T_\bfxi (D), \,\,T_\bfxi(A)\leq 2(T(F)+nT_\bfxi(D)).$$

\item Assume that $D=\frakP$ where the center of $\frakP$ is a singular point of $F=0$. Then there are two homogeneous polynomials $G,H\in \overkt[x_0,x_1,x_2]$ of degree two such that
$$
    \divs_\bfxi(G/H)=D+A
$$
where $A$ is a very simple divisor, $\supp(A)\cap (\supp(D')\cup \{\frakP\})=\emptyset$ and
$$T(G), T(H)\leq T(F)+nT_\bfxi(D),\,\,T_\bfxi(A)\leq (2n+4)T(F)+2n^2T_\bfxi(D).$$
\end{enumerate}
\end{prop}
\begin{proof}
$1.$
Suppose $\bfc=(c_0,c_1,c_2)$ is the center of $\frakP_i$ with respect to $\bfxi$. Without loss of generality, we assume that $c_2\neq 0$ and $\bfc=(c_0,c_1,1)$. Set $L_\lambda=x_0-c_0x_2-\lambda(x_1-c_1x_2)$.
Due to Lemma~\ref{lm:generallines}, for all but a finite number of $\lambda$, $L_\lambda=0$ intersects $F=0$ in $n-r$ distincet points other than the points in $\{\bfc\}\cup \calS_\bfxi(D')$. Let $\lambda'\in k$ be such that $L_{\lambda'}=0$ satisfies the above condition. Then
$$
\divs_{\bfxi}(L_{\lambda'})=\sum_{i=1}^{r} \frakP_{i}+A
$$
where $A$ is an effective divisor of degree $n-r$ and $\supp(A)\cap (\supp(D')\cup\{\frakP_1,\dots,\frakP_r\})=\emptyset$. It is clear that $A$ is very simple since $L_{\lambda'}=0$ intersects $F$ in $n-r$ distinct points other than $\bfc$. Finally, one easily sees that $T(L_{\lambda'})\leq T(\bfc)=T_\bfxi(D)$.  As the points in $T_\bfxi(A)$ are the intersection points of $F=0$ and $L_{\lambda'}=0$, $T_\bfxi(A)\leq 2(T(F)+nT_\bfxi(D))$ by Proposition~\ref{prop:intersection}.

$2.$ Suppose that $\bfc=(c_0,c_1,c_2)$ is the center of $\frakP$ with respect to $\bfxi$, and $\bfc$ is of multiplicity $r>0$. Since $\bfc$ is an ordinary singular point, there are exactly $r$ places of $\calR$ with $\bfc$ as the center with respect to $\bfxi$. Denote these $r$ places by $\frakP_1=\frakP, \dots, \frakP_r$. Without loss of generality, we may assume that $c_2\neq 0$ and $\bfc=(c_0,c_1,1)$.
Write
$$
    F=F_r(x_0-c_0x_2, x_1-c_1x_2)x_2^{n-r}+\dots+ F_n(x_0-c_0x_2,x_1-c_1x_2)
$$
where $F_i(y_0,y_1)$ is a homogeneous polynomial of degree $i$. Choose a projective change of coordinates $\calL$ with $M_\calL=\diag(B,1), B\in \GL_2(k)$
such that
$$
    F^{\calL}=\tilde{F}_r(x_0-\tilde{c}_0x_2,x_1-\tilde{c}_1x_2)x_2^{n-r}+\dots+\tilde{F}_n(x_0-\tilde{c}_0x_2,x_1-\tilde{c}_1x_2)
$$
satisfies that $\tilde{F}_r(1,0) \tilde{F}_r(0,1)\tilde{F}_n(1,0)\tilde{F}_n(0,1)\neq 0$, where $\tilde{F}_i=F_i((y_0,y_1)B)$ and $(\tilde{c}_0,\tilde{c}_1)=(c_0,c_1)B^{-1}$. By Lemma~\ref{lm:lineartransformation}, $T(F^\calL)\leq T(F)$.
 Denote
 $$\tilde{\bfxi}=(\tilde{\xi}_0,\tilde{\xi}_1,\tilde{\xi}_2)=\bfxi M_\calL^{-1}.$$
 Then  $\tilde{\bfc}=(\tilde{c}_0,\tilde{c}_1,1)=\bfc M_\calL^{-1}$ is the center of $\frakP_1$ with respect to $\tilde{\bfxi}$.  For $i=0,1$, write
 $$\tilde{\xi}_i/\tilde{\xi}_2=\tilde{c}_i+\alpha_i u^d+u^{d+1}\eta_i$$
  where $u$ is a local uniformizer of $\frakP_1$, $d\geq 1$, $\alpha_i\in \overkt$ not all zero,  and $\nu_{\frakP}(\eta_i)\geq 0$. Furthermore,
$$
   0=F^\calL(\tilde{\xi}_0/\tilde{\xi}_2,\tilde{\xi}_1/\tilde{\xi}_2,1)=u^{dr} \tilde{F}_r(\alpha_0,\alpha_1)+u^{dr+1}\beta
$$
 where $\nu_{\frakP_1}(\beta)\geq 0$. This implies that $\tilde{F}_r(\alpha_0,\alpha_1)=0$. Since $\tilde{F}_r(0,1)\tilde{F}_r(1,0)\neq 0$, $\alpha_0\alpha_1\neq 0$.
 Set $\bar{\alpha}=\alpha_0/\alpha_1$ and
$$\tilde{G}_\lambda=\bar{\alpha} (x_1-\tilde{c}_1x_2)x_2-(x_0-\tilde{c}_0x_2)x_2-\lambda (x_0-\tilde{c}_0x_2)(x_1-\tilde{c}_1x_2).$$
 Due to Lemma~\ref{lm:generallines}, for all but a finite number of $\lambda$, $\tilde{G}_\lambda$ intersects $F^\calL=0$ in $2n-r-1$ distinct points other than the points $\{\tilde{\bfc}\}\cup \calS_{\tilde{\bfxi}}(D')$. Let $A_\lambda$ be the very simple divisor consisting of the $2n-r-1$ places  whose centers with respect to $\tilde{\bfxi}$ are the intersection points of $\tilde{G}_\lambda=0$ and $F^\calL=0$ other than $\tilde{\bfc}$ respectively. Then $\supp(A_\lambda) \cap (\supp(D')\cup \{\frakP_1,\dots,\frakP_r\})=\emptyset$.
 We claim that for the above $\tilde{G}_\lambda$,
 $$
     \divs_{\tilde{\bfxi}}(\tilde{G}_\lambda)=\frakP_1+\sum_{i=1}^r \frakP_i +A_\lambda
 $$
 Note that
\begin{align*}
  \frac{ \tilde{G}_\lambda(\tilde{\bfxi})}{\tilde{\xi}_2^2}&=\bar{\alpha}\left(\alpha_1 u^d+u^{d+1} \eta_1\right)-(\alpha_0 u^d+u^{d+1}\eta_0)\\
  &-\lambda (\alpha_0 u^d+u^{d+1}\eta_0)(\alpha_1 u^d+u^{d+1} \eta_1)=u^{d+1}\gamma
\end{align*}
where $\nu_{\frakP_1}(\gamma)\geq 0$. This implies that $\ord_{\frakP_1}(\tilde{G}_\lambda(\tilde{\bfxi}))\geq d+1\geq 2$. Hence
$$\divs_{\tilde{\bfxi}}(\tilde{G}_\lambda)\geq \frakP_1+\sum_{i=1}^r \frakP_i +A_\lambda.$$
On the other hand, since $\deg(\divs_{\tilde{\bfxi}}(\tilde{G}_\lambda))=2n$, one has that
$$\divs_{\bfxi'}(\tilde{G}_\lambda)=\frakP_1+\sum_{i=1}^r \frakP_i +A_\lambda.$$
This proves our claim. Now set $G_\lambda=\tilde{G}_\lambda^{\calL^{-1}}$. As $\tilde{\xi}_2=\xi_2$, one sees that
$$
    \min_j\{\nu_{\frakP_i}(\tilde{\xi}_j)\}=\nu_{\frakP_i}(\tilde{\xi}_2)=\nu_{\frakP_i}(\xi_2)=\min_j\{\nu_{\frakP_i}(\xi_j)\}.
$$
This implies that
\begin{align*}
 \ord_{\frakP_i}(G_\lambda(\bfxi))&=\nu_{\frakP_i}\left(G_\lambda(\bfxi)\right)-2\min_i \{\nu_{\frakP_i}(\xi_i)\}\\
 &=\nu_{\frakP_i}\left(\tilde{G}_\lambda^{\calL^{-1}}(\bfxi)\right)-2\min_i \{\nu_{\frakP_i}(\tilde{\xi}_i)\}\\
 &=\nu_{\frakP_i}\left(\tilde{G}_\lambda(\tilde{\bfxi})\right)-2\min_i \{\nu_{\frakP_i}(\tilde{\xi}_i)\}=\ord_{\frakP_i}(\tilde{G}_\lambda(\tilde{\bfxi})).
\end{align*}
Therefore $\divs_\bfxi(G_\lambda)=\frakP_1+\sum_{i=1}^r \frakP_i +A_1$. Note that we can choose $\lambda\in k$. For such $\lambda$, one has that
$$
    T(G_\lambda)\leq T(\tilde{G}_\lambda)\leq 2T(\bfc)+T((\alpha_0,\alpha_1,0))\leq 2T(\bfc)+T(\tilde{F}_r).
$$
Since $\deg(\tilde{F}_r)=r\geq 2$, one sees that $T(\tilde{F}_r)\leq T(F)+(n-2)T(\bfc)$. This implies that $T(G_\lambda)\leq T(F)+nT(\bfc)$
and $$T(A_\lambda)\leq 2(2T(F)+nT(G_\lambda))\leq (2n+4)T(F)+2n^2T(\bfc).$$
Now applying $1.$ to the case that $D=\sum_{i=1}^r \frakP_i$ and $D'=A_\lambda$, one gets a linear homogeneous polynomial $L_1$ such that
$\divs_\bfxi(L_1)=\sum_{i=1}^r \frakP_i+C$, where $C$ is a very simple divisor satisfying that $\supp(C)\cap \{\supp(A_\lambda)\cup \{\frakP_1,\dots,\frakP_r\}\}=\emptyset$. Moreover $T(L_1)\leq T_\bfxi(\frakP_1)=T(\bfc)$. Let $L_2$ be a linear homogeneous polynomial in $k[x_0,x_1,x_2]$ such that $L_2=0$ intersects $F=0$ in $n$ distinct points other than the points in $T_\bfxi(A_\lambda+C)$.   For such $L_2$, one has that $\divs_\bfxi(L_2)$ is a very simple divisor satisfying that  $\supp(\divs_\bfxi(L_2))\cap \supp(A_1+C+D')=\emptyset$. Set $H=L_1L_2$ and $A=A_1-C-\divs_\bfxi(L_2)$. Then $T(H)\leq T(\bfc)$ and we obtain two polynomials $G,H$ as required. Note that $T_\bfxi(C)\leq 2(T(F)+nT_\bfxi(D))$ and $T_\bfxi(\divs_\bfxi(L_2))\leq 2T(F)$. Hence
$
    T(G_\lambda), T(H)\leq T(F)+nT(\bfc) \,\,\mbox{and}\,\,T_\bfxi(A) \leq (2n+4)T(F)+2n^2T(\bfc).
$
\end{proof}

\begin{define}
\label{def:adjoint}
Suppose that $F$ has only ordinary singularities, say $\bfq_1,\dots,\bfq_\ell$, and $r_i$ is the multiplicity of $\bfq_i$. Suppose further that for each $i=1,\dots,\ell$, $\frakQ_{i,1}, \dots, \frakQ_{i,r_i}$ are all places of $\calR$ with $\bfq_i$ as center with respect to $\bfxi$.
Set
$$
E_\bfxi=\sum_{i=1}^\ell (r_i-1)\sum_{j=1}^{r_i}\frakQ_{i,j}.
$$
A homogeneous polynomial $G$ such that $\divs_\bfxi(G)\geq E_\bfxi$ is called an {\em adjoint} of $F$.
\end{define}

We have the following two corollaries of Proposition~\ref{prop:simplification}.
\begin{corollary}
\label{cor:simplification1}
Suppose that $D$ is a simple and effective divisor in $\calR$. Let $D'$ be a divisor in $\calR$. Then there is a homogeneous polynomial $G$ of degree  not greater than $\deg(D)+(n-1)^2/2$ such that
$$
   \divs_\bfxi(G)=D+E_\bfxi+A
$$
where $A$ is a very simple and effective divisor of degree not greater than
$$\deg(D)(n-1)+n(n-1)^2/2$$
such that $\supp(A)\cap (\supp(D+E_\bfxi)\cup \supp(D'))=\emptyset$ and $T_\bfxi(A)\leq 2(T(F)+nT_\bfxi(D+E_\bfxi))$. Moreover
$$T(G)\leq \left(\deg(D)+(n-1)^2/2\right)T_\bfxi(D+E_\bfxi).$$
\end{corollary}
\begin{proof}
Denote $\mu=\deg(D)+\sum_{i=1}^\ell (r_i-1)$ where $r_i$ is given as in Definition~\ref{def:adjoint}. Note that $\sum_{i=1}^\ell (r_i-1)\leq (n-1)^2/2$, $\mu\leq \deg(D)+(n-1)^2/2$.
Write
$$D+E_\bfxi=\sum_{i=1}^{\deg(D)}\frakP_i +\sum_{s=\deg(D)+1}^{\mu} D_s$$
where the centers of $\frakP_i$ is a simple point of $F=0$,  $D_s=\sum_{j=1}^{r_i}\frakQ_{i,j}$ for some $1\leq i \leq \ell$. Applying succesively Proposition~\ref{prop:simplification} to $\frakP_i$ and $D_s$, one obtains $\mu$ linear homogeneous polynomials $L_1,\dots, L_\mu$ such that $\divs_\bfxi(L_i)=\frakP_i+A_i$ if $i\leq m$, or $\divs_\bfxi(L_i)=D_i+A_i$ if $i>m$, where $A_i$ is a very simple and effective divisor such that $$\supp(A_i)\cap (\supp(D')\cup \supp(D+E_\bfxi+A_1+\dots+A_{i-1}))=\emptyset.$$
Set $G=\sum_{i=1}^\mu L_i$ and $A=\sum_{i=1}^\mu A_i$. Then one has that
$$\divs_\bfxi(G)=D+E_\bfxi+A.$$
Moreover by Proposition~\ref{prop:simplification}, $T(L_i)\leq T_\bfxi(D+E_\bfxi)$ for all $i=1,\dots,\mu$ and then Proposition~\ref{prop:height2} implies that $T(G)\leq \mu T_\bfxi(D+E_\bfxi)$. It is obvious that $T_\bfxi(A)$ is not greater than $2(T(F)+nT_\bfxi(D+E_\bfxi))$ because so is $T_\bfxi(A_i)$ for all $i=1,\dots,\mu$.
\end{proof}
\begin{corollary}
\label{cor:simplification2}
Suppose that $D, D'$ are two divisors in $\calR$. Then there are two homogeneous polynomials $G,H$ of the same degree  $\leq 2\deg(D^{+}+D^{-})$ such that
$
  \hat{D} =\divs_\bfxi(G/H)+D
$
is very simple and $\supp(\divs_\bfxi(G/H)+D)\cap \supp(D')=\emptyset$. Moreover  $\deg(\hat{D}^{+}), \deg(\hat{D}^{-})\leq  2n(\deg(D^{+}+D^{-}))$ and
\begin{align*}
 T_\bfxi(\divs_\bfxi(G/H)+D)& \leq (2n+4)T(F)+2n^2T_\bfxi(D)\\
T(G), T(H)&\leq \deg(D^{+}+D^{-})(T(F)+nT_\bfxi(D)),
\end{align*}
where $n=\deg(F)$.
\end{corollary}
\begin{proof}
We first show the case that $-D$ is effective. Denote $\mu=\deg(-D)$ and write
$$
    -D=\sum_{i=1}^s \frakP_i + \sum_{i=s+1}^\mu \frakQ_i
$$
where the center of $\frakP_i$ (resp. $\frakQ_j$) with respect to $\bfxi$ is a simple (resp. singular) point of $F=0$.  Applying Proposition~\ref{prop:simplification} to $\sum_{i=1}^s \frakP_i$ yields a homogenenous polynomial $G_0$ of degree $s$ such that
$\divs_\bfxi(G_0)=\sum_{i=1}^s \frakP_i +A_0$ where $A_0$ is a very simple and effective divisor such that $\supp(A_0)\cap (\supp(D)\cup \supp(D'))=\emptyset$.  Moreover $\deg(A_0)=ns-\deg(\sum_{i=1}^s\frakP_i)$. Construct $s$ linear homogeneous polynomials $L_1,\dots,L_s$ in $k[x_0,x_1,x_2]$ such that $\divs_\bfxi(L_1\cdots L_s)$ is very simple and $\supp(\divs_\bfxi(L_1\cdots L_s))\cap (\supp(\divs_\bfxi(G_0)\cup \supp(D'))=\emptyset$.  It is easy to see that $\deg(\divs_\bfxi(L_1\cdots L_s))=ns$. Set $H_0=L_1\cdots L_s$. By Proposition~\ref{prop:simplification} again, one obtains  $\mu-s$ pairs $(G_1, H_1), \dots, (G_{\mu-s}, H_{\mu-s})$ of homogeneous polynomials of degree two such that $\divs_\bfxi(G_i/H_i)=\frakQ_i+A_i$ where $A_i$ is a very simple divisor such that
$$\supp(A_i)\cap (\supp(D')\cup \supp(D+A_0+\cdots+A_{i-1}))=\emptyset$$
and $\deg(A_i^{+})=2n-\deg(\frakQ_i)$, $\deg(A_i^{-})=2n$.
Set $\tilde{G}=G_0G_1\cdots G_{\mu-s}$ and $\tilde{H}=H_0H_1\cdots H_{\mu-s}$. Then
$$\hat{D}=\divs_\bfxi(\tilde{G}/\tilde{H})+D=A_0+A_1+\dots+A_{\mu-s}$$
which is very simple. It is clear that $\deg(\tilde{G})=\deg(\tilde{H})\leq 2\deg(-D)$, and by Proposition~\ref{prop:simplification}
$$T(\tilde{G})\leq T(G_0)+(\mu-s)T(G_i)\leq \mu (T(F)+nT_\bfxi(D)).$$
Similarly, $T(\tilde{H})\leq \mu (T(F)+nT_\bfxi(D))$. Furthermore, one has that
$$T_\bfxi\left(\divs_\bfxi(\tilde{G}/\tilde{H})+D\right)\leq (2n+4)T(F)+2n^2T_\bfxi(D)$$
and
$$
   \deg(\hat{D}^{+})=(2\mu-s)n-\deg(-D)\leq 2n\mu, \, \deg(\hat{D}^{-})=(2\mu-s)n\leq 2n\mu.
$$
For the general case, write $D=D^{+}-D^{-}$. The previous discussion implies that we can obtain $\tilde{G}_i, \tilde{H}_i$ such that
$
   \divs_\bfxi(\tilde{G}_1/\tilde{H_1})-D^{+}
$ and $\divs_\bfxi(\tilde{G}_2/\tilde{H_2})-D^{-}$ are very simple. Moreover
\begin{align*}
    \supp\left(\divs_\bfxi (\tilde{G}_1/\tilde{H_1})-D^{+}\right)\bigcap \supp\left(\divs_\bfxi(\tilde{G}_2/\tilde{H_2})-D^{-}\right)&=
\emptyset,\\
   \supp\left(\divs_\bfxi(\tilde{G}_2/\tilde{H_2})-D^{-}+\divs_\bfxi(\tilde{G}_1/\tilde{H_1})-D^{+}\right)\bigcap \left(\supp(D)\cup \supp(D')\right)&=\emptyset.
\end{align*}
Set $G=\tilde{G}_2\tilde{H}_1$ and $H=\tilde{G}_1\tilde{H}_2$. Then $\divs_\bfxi(G/H)+D$ satisfies the required condition. Furthermore
$\deg(\hat{D}^{+}), \deg(\hat{D}^{-})\leq 2n(\deg(D^{+}+D^{-}))$ and
\begin{align*}
  T(G), T(H)&\leq \deg(D^{+}+D^{-})(T(F)+nT_\bfxi(D)),\\
   T_\bfxi(\divs_\bfxi(G/H)+D)&\leq (2n+4)T(F)+2n^2T_\bfxi(D).
\end{align*}
\end{proof}
Now we are ready to pove the main results of this section. Let us start with two lemmas.
\begin{lemma}
\label{lm:riemannrochspace}
Suppose that $F=0$ has only ordinary singularities and $D$ is a divisor in $\calR$. Let $H$ be a homogeneous polynomial in $\overkt[x_0,x_1,x_2]$ such that $\divs_\bfxi(H)=D^{+}+E_\bfxi+A$, where $A$ is an effective divisor. Then
\begin{equation}
\label{eq:riemannrochspace}
    \frakL(D)=\left\{\left.\frac{G(\bfxi)}{H(\bfxi)} \,\right |  \,\begin{array}{c} \mbox{$G$ are homogeneous polynomials of $\deg(H)$}\\
                                                \mbox{with $\divs_\bfxi(G)\geq D^{-}+E_\bfxi+A$}
                                                \end{array}\right\}.
\end{equation}
\end{lemma}
\begin{proof}
Note that $\divs(G(\bfxi)/H(\bfxi))=\divs_\bfxi(G)-\divs_\bfxi(H)$.
It is obvious that the right hand side of (\ref{eq:riemannrochspace}) is a subspace of $\frakL(D)$.
Suppose that $h\in \frakL(D)\setminus\{0\}$, i.e. $D'=\divs(h)+D$ is effective. Then
$$\divs_\bfxi(H)+\divs(h)=D^{-}+D'+E_\bfxi+A.$$
By the Residuce Theorem (see page of \cite{fulton}), there is a homogeneous polynomial $G$ of degree $\deg(H)$ such that
$$
     \divs_\bfxi(G)=D^{-}+D'+E_\bfxi+A\geq D^{-}+E+A.
$$
One sees that
$$
   \divs(hH(\bfxi)/G(\bfxi))=\divs(h)+\divs_\bfxi(H/G)=\divs(h)+\divs_\bfxi(H)-\divs_\bfxi(G)=0.
$$
Thus $hH(\bfxi)/G(\bfxi)\in \overkt$, i.e. $h$ belongs to the right hand side of (\ref{eq:riemannrochspace}).
\end{proof}

\begin{lemma}
\label{lm:linearsystem}
Assume that $M=(a_{i,j})$ is an $l \times m $ matrix with $a_{i,j}\in \overkt$.  Assume further that for each place $\frakp$ of $k(t, a_{1,1},\dots,a_{l,m})$,
$
   -\nu_\frakp(a_{i,j})\leq m_\frakp
$
where $m_\frakp\geq 0$ and for all but finite number of $\frakp$, $m_\frakp\neq 0$.
Then there is a basis $B$ of the solution space of $MY=0$ satisfying that
$$T(\bfb)\leq \frac{ \min \{l,m\} \sum_{\frakp} m_\frakp}{[k(t,a_{1,1},\dots,a_{l,m}):k(t)]}$$
for all $\bfb\in B$.
\end{lemma}
\begin{proof}
Assume that $r=\rank(M)$. Then $r\leq \min\{l,m\}$. Without loss of generality, we may assume the first $r$-rows of $M$ are linearly independent and denote by $\tilde{M}$ the matrix formed by them.
Then the solution space of $\tilde{M}Y=0$ is the same as that of $MY=0$.  Hence it suffices to consider the system $\tilde{M}Y=0$. We may further assume that the matrix  $\tilde{M}_1$ formed by the first $r$-columns of $\tilde{M}$ is invertible. For every $i=1,\dots,r$ and $j=r+1,\dots,m$, set $d_{i,j}$ to be the determinant of the matrix obtained from $\tilde{M}_1$ by replacing the $i$-th column of $\tilde{M}_1$ by the $j$-th column of $\tilde{M}$. For each $j=r+1,\cdots, m$, denote
$$
   \bfc_j=(d_{1,j},\dots,d_{r,j}, 0,\dots, 0,\underbrace{\det(M_1)}_{j},0,\dots,0)^t
$$
where $(\cdot)^t$ denotes the transpose of a vector.
 Then by Cramer's rule, the $\bfc_j$ are solutions of $\tilde{M}Y=0$ and thus they form a basis of the solution space of $\tilde{M}Y=0$. Note that $d_{i,j}$ as well as $\det(\tilde{M}_1)$ is an integer combination of the monomials in the entries of $\tilde{M}$ of total degree $r$. So for all $i=1,\dots,r$ and $j=r+1,\dots,m$,
 $$
    -\nu_\frakp(\tilde{M}_1), -\nu_\frakp(d_{i,j})\leq  r m_\frakp \leq \min\{l,m\} m_\frakp $$
 where $\frakp$ is a place of $k(t,a_{1,1},\dots,a_{l,m})$. This together Remark~\ref{rem:heights} implies the lemma.
\end{proof}

\begin{theorem}
\label{thm:riemann-roch1}
Suppose that $F=0$ has only ordinary singularities.
Let $D$ be a divisor in $\calR$. Denote $\mu=\deg(D^{+}+D^{-})$ and
$N=\max\{T_\bfxi(D), T(F)\}.$ Then there is a $\overkt$-basis $B$ of $\frakL(D)$ such that every element of $B$ can be represented by $G(\bfxi)/H(\bfxi)$ where $G,H$ are two homogeneous polynomials of the same degree not greater than $\leq 2(n+1)\mu+(n-1)^2/2$ and
$$T(G), T(H)\leq  4n^5(n+1)^3(2\mu+(n-1)/2)^3 N.$$
\end{theorem}
\begin{proof}
By Corollary~\ref{cor:simplification2}, there are two homogeneous polynomials $G_1,H_1$ of the same degree $\leq 2\mu$ such that $\hat{D}=\divs_\bfxi(G_1/H_1)+D$ is very simple. Moreover
\begin{align*}
     T(G_1), T(H_1) & \leq \mu (n+1)N,\\
     T_\bfxi(\hat{D}) \leq (2n+4)T(F)+2n^2T_\bfxi(D) &\leq (2n^2+2n+4)N,\\
      \deg(\hat{D}^{+}), \deg(\hat{D}^{-})&\leq 2n\mu.
\end{align*}
Due to Corollary~\ref{cor:simplification1}, there is a homogeneous polynomial $G_2$ of degree not greater than
$2n\mu+(n-1)^2/2$
such that $\divs_\bfxi(G_2)=\hat{D}^{+}+E_\bfxi+A$, where $A$ is a very simple and effective divisor and $\supp(A)\cap \supp(\hat{D}^{-})=\emptyset$.
Moreover
$$T(G_2)\leq (2n\mu+(n-1)^2/2) T_\bfxi(\hat{D}^{+}+E_\bfxi)\leq (2n\mu+(n-1)^2/2)(2n^2+2n+4)N$$
and
$$
    T_\bfxi(A)\leq 2(T(F)+nT_\bfxi(\hat{D}^{+}+E_\bfxi))\leq 2(2n^3+2n^2+4n+1)N.
$$
Denote $d=\deg(G_2)$. By Lemma~\ref{lm:riemannrochspace}, to compute $\frakL(D)$, it suffices to compute all homogeneous polynomials $H_2$ of degree $d$ satisfying that 
$$
      \divs_\frakP(H_2)\geq \hat{D}^{-}+E_\bfxi+A.
$$
Assume that
$$
      H_2=\sum_{i=0}^d \sum_{j=0}^{d-i} c_{i,j} x_0^i x_1^j x_2^{d-i-j}
$$
where $c_{i,j}$ are indeterminates. There are $(d+1)(d+2)/2$ indeterminates in total. For each $\frakP\in \supp(\hat{D}^{-}+A)$, $\divs_\bfxi(H_2)\geq \frakP$ if and only if the center of $\frakP$ with respect to $\bfxi$ is a zero of $H_2$. This imposes $\deg(\hat{D}^{-}+A)$ linear constraints on $H_2$. At the same time, $\divs_\bfxi(H_2)\geq (r_i-1)\sum_{j=1}^{r_i}\frakQ_{i,j}$ if and only if the center of $\frakQ_{i,1}$ with respect to $\bfxi$ is a common zero of
$$\frac{\partial^{j_0+j_1+j_2}(H_2)}{\partial x_0^{j_0}x_1^{j_1}x_2^{j_2}}$$
for all nonnegative integers $j_0,j_1,j_2$ satisfying that $j_0+j_1+j_2=r_i-2$, where $\frakQ_{i,j}$ is as in Definition~\ref{def:adjoint}. This imposes $r_i(r_i-1)/2$ linear constraints on $H_2$. So
there are totally $\deg(\hat{D}^{-}+A)+\deg(E_\bfxi)/2$ linear constraints on $H_2$. The problem of finding $H_2$ is reduced to that of solving the system
$M Y=0$, where $Y$ is a vector with indeterminates entries and $M$ is a   $(\deg(\hat{D}^{-}+A)+\deg(E_\bfxi)/2)\times (d+1)(d+2)/2$ matrix. Denote by $\bfc_\frakP=(c_{0,\frakP}, c_{1,\frakP}, c_{2,\frakP})$ the center of $\frakP$ in $\supp(\hat{D}^{-}+E_\bfxi+A)$.
Then the entries in the same row of $M$ are monomials of total degree $\leq d$ in $c_{0,\frakP}, c_{1,\frakP}, c_{2,\frakP}$  for some $\frakP$ in $\supp(\hat{D}^{-}+E_\bfxi+A)$. Without loss of generality, we may assume that one of $c_{0,\frakP}, c_{1,\frakP}, c_{2,\frakP}$ is 1.  Let $R$ be a finite extension of $k(t)$ containing all $c_{i,\frakP}$. For each place $\frakp$ of $R$, set
$$
    m_\frakp=d \sum_{\frakP\in \supp(\hat{D}^{-}+E_\bfxi+A)}\max\{-\nu_\frakp(c_{0,\frakP}), -\nu_\frakp(c_{1,\frakP}), -\nu_\frakp(c_{2,\frakP})\}.
$$
Since $\max\{-\nu_\frakp(c_{0,\frakP}), -\nu_\frakp (c_{1,\frakP}), -\nu_\frakp(c_{2,\frakP})\}\geq 0$  for all $\frakP$, $m_\frakp\geq 0$ and
\begin{align*}
     -\nu_\frakp(a_{i,j})\leq d \max_{\frakP\in\supp(\hat{D}^{-}+E_\bfxi+A)}\max\{-\nu_\frakp(c_{0,\frakP}), -\nu_\frakp(c_{1,\frakP}), -\nu_\frakp(c_{2,\frakP})\} \leq m_\frakp
\end{align*}
where $M=(a_{i,j})$. Note that
$$
   \deg(\hat{D}^{-}+E_\bfxi+A)\leq \deg(\divs_\bfxi(H_2))=nd.
$$
Applying Lemma~\ref{lm:linearsystem} to $M$ yields that
\begin{align*}
T(H_2)& \leq \deg(\hat{D}^{-}+E_\bfxi+A)\sum_{\frakp}\frac{m_\frakp}{[R:k(t)]}\\
& \leq nd^2\sum_{\frakP} \sum_{\frakp}\frac{\max\{-\nu_\frakp(c_{0,\frakP}), -\nu_\frakp(c_{1,\frakP}), -\nu_\frakp(c_{2,\frakP})\}}{[R:k(t)]}\\
&\leq nd^2\sum_{\frakP} T(\bfc_\frakP) \leq nd^2\deg(\hat{D}^{-}+E_\bfxi+A) \max_{\frakP} T(\bfc_\frakP)\\
&\leq n^2d^3 T_\bfxi(\hat{D}^{-}+E_\bfxi+A)\leq 2n^2d^3(2n^3+2n^2+4n+1)N \\
&\leq 2n^5(2\mu+(n-1)/2)^3(2n^3+2n^2+4n+1)N.
\end{align*}
The last inequality holds because
$$d\leq 2n\mu+(n-1)^2/2\leq n(2\mu+(n-1)/2).$$
Set $G=H_2G_1$ and $H=G_2H_1$. Then
\begin{align*}
\deg(G)&=\deg(H)\leq 2(n+1)\mu+(n-1)^2/2,\\
     T(G), T(H)&\leq T(H_2)+T(G_1)<T(H_2) + \mu (n+1)N\\
       &\leq 2n^5(2\mu+(n-1)/2)^3(2n^3+2n^2+4n+2)N\\
       &\leq 4n^5(n+1)^3(2\mu+(n-1)/2)^3 N.
\end{align*}
 \end{proof}

 Next, we consider the case that $F=0$ may have non-ordinary singular points. Let $\calC(\bfxi)$ be a plane projective model of $\calR$. Suppose that $\calC(\bfxi)$ is defined by $F_0$ and for $i=1,\dots,s$, $F_i$ is the quadratic transformation of $F_{i-1}$ under  the quadratic transformation $\calQ_{\bfc_{i-1}}$, where $\bfc_{i-1}$ is a singular point of $F_{i-1}=0$.  Denote $n=\deg(F_0)$ and set $\bfxi_0=\bfxi$, $\bfxi_{i+1}=\calQ_{\bfc_{i}}^{-1}(\bfxi_{i})$ and
\begin{equation}
\label{eq:seqtransformation}
    N_i=2^{\frac{i(i-1)}{2}}n^{i}\max\{8nT(F_0),T_{\bfxi_0}(D)\}.
\end{equation}
\begin{prop}
\label{prop:seqtransformation}
Let $D$ be a divisor in $\calR$ and the notations $F_i, \bfxi_i, N_i$ as above.  One has that
\begin{enumerate}
\item
$\tdeg(F_i)\leq n 2^{i}-2^{i+1}+2$;
\item
$ T_{\bfxi_i}(D), T(F_i)\leq N_i$.
\end{enumerate}
\end{prop}
\begin{proof}
1. Set $n_i=\deg(F_i)$. Since every  $\bfc_i$ is a singular point, one has that $n_i\leq  2n_{i-1}-2$. This implies $n_i\leq 2^i n -2^{i+1}+2$.

2. Denote by $S_i$ the maximum of the heights of singular points of $F_i=0$. We first prove by induction on $i$ that $T(F_i), S_i\leq N_i$ for all $i=0,\dots,s$. Note that $n_i+1<2^i n$ for all $i=1,\dots,s$. Since $N_0=8nT(F_0)$, it is clear that $T(F_0)<N_0$ and $S_0<4nT(F_0)<N_0$ by Corollary~\ref{cor:singularity}.
Now assume that $T(F_i), S_i\leq N_i$ for $i=\ell\geq 0$. Consider the case $i=\ell+1$. By Corollary~\ref{cor:qtransformation} and induction hypothesis, one has that
$$
  T(F_{\ell+1})  \leq  T(F_\ell)+n_{\ell}S_\ell\leq (1+n_\ell)N_\ell<2^\ell n N_\ell=N_{\ell+1}.
$$
Note that $n_0=n>2$ as the curve $F_0=0$ has singularities. One sees that  $4S_\ell \leq 4N_\ell < 2^\ell n N_\ell=N_{\ell+1}$ if $\ell>0$ and
$$4S_0< 16nT(F_0)<8n^2 T(F_0)\leq N_1.$$
Consequently, $4S_j<N_{j+1}$ for all $j\geq 0$. On the other hand, one has already seen that
$T(F_\ell)+n_\ell S_\ell <N_{\ell+1}.$
By Corollary~\ref{cor:qtransformation} again,
\begin{align*}
 S_{\ell+1}\leq \max\{4S_\ell, T(F_{\ell})+n_{\ell}S_\ell\}<N_{\ell+1}.
\end{align*}
 For the divisor $D$,  it is obvious that $T_{\bfxi_0}(D)\leq N_0$. Suppose that $T_{\bfxi_i}(D)\leq N_i$ for $i=\ell\geq 0$. By Corollary~\ref{cor:qtransformation} and the induction hypothesis,
 $$
     T_{\bfxi_{\ell+1}}(D)\leq \max\{2(S_\ell+N_\ell), T(F_\ell)+n_\ell T_{\bfxi_\ell}(D)\}\leq N_{\ell+1}.
 $$
 \end{proof}
 \begin{notation}
 \label{not:steps}
Let $F$ be the defining polynomial of $\calC(\bfxi)$.
 Denote by $s(\bfxi)$ the number of quadratic transformations such that $\calC(\tilde{\bfxi})$ has only ordinary singularities, where $\tilde{\bfxi}$ is the image of $\bfxi$ under these quadratic transformations. By Theorem 2 in Chapter 7 of \cite{fulton}, $s(\bfxi)$ can be chosen to be an integer not greater than
$$
m+\frac{(n-1)(n-2)}{2}-\sum \frac{r_\bfc(r_\bfc-1)}{2}\leq \frac{(n-1)(n-2)}{2}
$$
where $n=\deg(F)$, $m$ is the number of non-ordinary singularities of $F=0$, $\bfc$ ranges over all singularities of $F=0$ and $r_\bfc$ is the multiplicity of $\bfc$. 
 \end{notation}

 \begin{theorem}
\label{thm:riemann-roch2}
Let $D$ be a divisor in $\calR$. Denote
$$n=\deg(F),  s=s(\bfxi), \mu=\deg(D^{+}+D^{-}).$$
Then there is a $\overkt$-basis $B$ of $\frakL(D)$ such that every element of $B$ can be represented by $G(\bfxi)/H(\bfxi)$ where $G,H$ are two homogeneous polynomials of the same degree not greater $2^{2s+1}(n+1)(\mu+2^{s-2}n)$ and
$$T(G), T(H)\leq 2^{\frac{s^2}{2}+\frac{15s}{2}+5}n^{s+5}(n+1)^3(\mu+2^{s-2}n)^3 \max\{8nT(F), T_\bfxi(D)\}.$$
\end{theorem}
\begin{proof}
If $F=0$ has only ordinary singularities, i.e. $s=0$, then the assertion is clear by Theorem \ref{thm:riemann-roch1}. Suppose $F=0$ has non-ordinary singularities, i.e. $s\geq 1$.
Let $\tilde{\bfxi}$ be the image of $\bfxi$ under $s$ quadratic transformations $\calQ^{-1}_{\bfc_0}, \dots, \calQ^{-1}_{\bfc_{s-1}}$ such that $\calC(\tilde{\bfxi})$ has only ordinary singularities.
Let $\tilde{F}$ be the defining polynomial of $\calC(\tilde{\bfxi})$ and set
$$\kappa=2^{s(s-1)/2}n^s \max\{8nT(F), T_\bfxi(D)\}.$$
Then by Proposition~\ref{prop:seqtransformation}
$$
\tilde{n}=\deg(\tilde{F})\leq 2^s(n-2)+2 \,\,\mbox{and}\,\,
T_{\tilde{\bfxi}}(D), T(\tilde{F})\leq \kappa.
$$
By Theorem~\ref{thm:riemann-roch1}, there is a $\overkt$-basis $B$ of $\frakL(D)$ satisfying that
each element in $ B$ can be represented by $\tilde{G}(\tilde{\bfxi})/\tilde{H}(\tilde{\bfxi})$ where  $\tilde{G},\tilde{H}$ are homogeneous polynomials of degree not greater than $2(\tilde{n}+1)\mu+(\tilde{n}-1)^2/2$ and
$$
   T(\tilde{G}), T(\tilde{H}) \leq   4\tilde{n}^5(\tilde{n}+1)^3 (2\mu+(\tilde{n}-1)/2)^3 \kappa.
$$
It remains to represent elements of $B$ in terms of $\bfxi$. We use the same notations as in the proof of Propositoin~\ref{prop:seqtransformation}. Let $\bfxi_0=\bfxi$ and $\bfxi_i=\calQ_{\bfc_{i-1}}^{-1}(\bfxi_{i-1})$. Denote by $n_i$ the degree of the defining polynomial of $\calC(\bfxi_i)$ and $S_i$ the maximum of the heights of singular points of $\calC(\bfxi_i)$. Let $G_s=\tilde{G}$ and
$G_{i-1}=G_i(\calQ_{\bfc_{i-1}}^{-1}((x_0,x_1,x_2)))$
for all $i=1,\dots,s.$ One sees that $\deg(G_i)=2^{s-i}\deg(\tilde{G})$. By Lemmas~\ref{lm:lineartransformation} and~\ref{lm:standardtransformation},
$$
T(G_{i-1})\leq T(\bar{G}_i)+\deg(\bar{G}_i)T(\bfc_{i-1})=T(G_i)+2\deg(G_i)T(\bfc_{i-1})
$$
where $\bar{G}_i=G_i(\calQ^{-1}((x_0,x_1,x_2)))$.
From the proof of Proposition~\ref{prop:seqtransformation}, we have
\begin{align*}
    T(G_0)&\leq T(\tilde{G})+\sum_{i=0}^{s-1}2\deg(G_{i+1})T(\bfc_i)
    \leq T(\tilde{G})+(\sum_{i=0}^{s-1}2^{s-i})\deg(\tilde{G})N_{s-1}\\
    &\leq T(\tilde{G})+2^{s+1}\deg(\tilde{G})N_{s-1},
 \end{align*}
where $N_{s-1}$ is given as in (\ref{eq:seqtransformation}). Note that $\tilde{n}\leq n2^s$. One has that
\begin{align*}
\deg(G_0)& \leq 2^s \deg(\tilde{G}) \leq 2^s (2(\tilde{n}+1)\mu+(\tilde{n}-1)^2/2)  \\
&\leq 2^s (\tilde{n}+1)(2\mu+(\tilde{n}-1)/4)\leq 2^{2s+1}(n+1)(\mu+n2^{s-2});\\
T(G_0)&\leq  4\tilde{n}^5(\tilde{n}+1)^3 \left(2\mu+\frac{\tilde{n}-1}{2}\right)^3 \kappa + 2^{s+1}\left(2(\tilde{n}+1)\mu+\frac{(\tilde{n}-1)^2}{2}\right)N_{s-1} \\
&\leq 4\tilde{n}^5(\tilde{n}+1)^3 \left(2\mu+\frac{\tilde{n}-1}{2}\right)^3 \kappa + 2^{s+1}(\tilde{n}+1)\left(2\mu+\frac{\tilde{n}-1}{2}\right)\kappa\\
&\leq \left(4\tilde{n}^5\left(2\mu+\frac{\tilde{n}-1}{2}\right)+2^{s+1}\right)(\tilde{n}+1)^3 \left(2\mu+\frac{\tilde{n}-1}{2}\right)^2 \kappa \\
&\leq \left(4n^52^{5s}\left(2\mu+\frac{n2^s-1}{2}\right)+2^{s+1}\right)(\tilde{n}+1)^3 \left(2\mu+\frac{\tilde{n}-1}{2}\right)^2 \kappa \\
&\leq 8n^52^{5s}(\mu+n2^{s-2})(\tilde{n}+1)^3 \left(2\mu+\frac{\tilde{n}-1}{2}\right)^2 \kappa \\
&\leq 2^{8s+5}n^5(n+1)^3(\mu+n2^{s-2})^3 \kappa \\
&\leq 2^{\frac{s^2}{2}+\frac{15s}{2}+5}n^{s+5}(n+1)^3(\mu+2^{s-2}n)^3 \max\{8nT(F), T_\bfxi(D)\}.
\end{align*}
Similarly, we obtain bounds for $\deg(H_0)$ and $T(H_0)$.
\end{proof}
\section{Heights on plane algebraic curves}
\label{sec:heights}
Let $f\in k[x_0,x_1]$ be an irreducible polynomial over $k$ and $a,b\in k(t)\setminus \{0\}$ satisfy $f(a,b)=0$, i.e. $(a,b)$ is a rational parametrization of $f=0$. The result on parametrization (see \cite{sendra-winkler} for instance) implies that
$$
    \deg(a)=m \deg(f,x_1),\,\,\deg(b)=m\deg(f,x_0).
$$
In other words, $T(a)\deg(f,x_0)=T(b)\deg(f,x_1)$. A similar relation holds for points in algebraic curves defined over $\overkt$, i.e  there is a constant $C$ only depending on $f$ such that  if $(a,b)$ is a point of $f(x_0,x_1)=0$ with coordinates in $\overkt$ then
$$
      \deg(f,x_0)T(a)-C \leq \deg(f,x_1)T(b)\leq \deg(f,x_0)T(a)+C.
$$
This is a special case of a general result for points in complete nonsingular varieties over a field with valuations.
 In the case of algebraic curves defined over  $\overkt$, Eremenko in 1999 presented another  proof which actually provides a procedure to find $C$ explicitly. In this section, we shall present an explicit formula for $C$ following Eremenko's proof.

 \subsection{Heights  on plane projective curves}
Throughout this subsection, $f$ is an irreducible polynomial in $\overkt[x_0,x_1]$ and $\calR$ is the algebraic function field over $\overkt$ associated to $f$. Let us start with a refinement of Lemma 1 of \cite{eremenko}.
\begin{lemma}
\label{lm:pointsinequality}
Assume  that $f\in \overkt[x_0,x_1]$ is irreducible over $\overkt$ and $\alpha,\beta\in \calR\setminus \overkt$ satisfying $f(\alpha,\beta)=0$.  If $\divs(\alpha)^{-} \leq \divs(\beta)^{-}$, then for every place $\frakP$ of $\calR$ with $\nu_{\frakP}(\beta)\geq 0$, we have that
$$
T(\pi_\frakP(\alpha))\leq T(\pi_\frakP(\beta))+T(f).
$$
\end{lemma}
\begin{proof}
Since $\divs(\alpha)^{-}\leq \divs(\beta)^{-}$, by Proposition 2 of \cite{eremenko}, $f$ can be written to be of the form
$$
     f= x_0^n+a_{n-1}(x_1)x_0^{n-1}+\dots+a_1(x_1)x_0+a_0(x_1),
$$
where $a_i\in \overkt[x_1]$ with $\deg(a_i)\leq n-i$. Write $a_i=\sum_{j=0}^{n-i} a_{i,j} x_1^j$ with $a_{i,j}\in \overkt$. Let $R$ be a finite extension of $k(t)$ containing all $a_{i,j}$ and $\pi_\frakP(\alpha), \pi_\frakP(\beta)$. Suppose that $\frakp$ is a place of $R$. Then
\begin{align*}
   \nu_\frakp(\pi_\frakP(\alpha^n))&=\nu_\frakp\left(-\sum_{i=0}^{n-1} \sum_{j=0}^{n-i}a_{i,j}\pi_\frakP(\beta)^j \pi_\frakP(\alpha)^i\right)\\
   &\geq \min_{0\leq i \leq n-1,0\leq j\leq n-i}\{\nu_\frakp(a_{i,j})+j\nu_\frakp(\pi_\frakP(\beta))+i\nu_\frakp(\pi_\frakP(\alpha))\}\\
   &=\nu_\frakp(a_{i',j'})+j'\nu_\frakp(\pi_\frakP(\beta))+i'\nu_\frakp(\pi_\frakP(\alpha))
\end{align*}
for some $0\leq i' \leq n-1, 0\leq j'\leq n-i'$.
Equivalently,
$$
    \nu_\frakp(\pi_\frakP(\alpha))\geq \frac{1}{n-i'}\nu_\frakp(a_{i',j'})+\frac{j'}{n-i'}\nu_\frakp(\pi_\frakP(\beta)).
$$
Therefore
\begin{align*}
    \max\{0, -\nu_\frakp(\pi_\frakP(\alpha))\} &\leq \max\left\{0,-\frac{\nu_\frakp(a_{i',j'})}{n-i'}-\frac{j'\nu_\frakp(\pi_\frakP(\beta))}{n-i'}\right\}\\
    & \leq \max\left\{0,-\nu_\frakp(a_{i',j'})\right\}+\max\left\{0,-\nu_\frakp(\pi_\frakP(\beta))\right\}\\
    &\leq \max_{i,j}\left\{0,-\nu_\frakp(a_{i,j})\right\}+\max\left\{0,-\nu_\frakp(\pi_\frakP(\beta))\right\}.
\end{align*}
This implies that $T(\pi_\frakP(\alpha))\leq T(\pi_\frakP(\beta))+T(f)$.
\end{proof}
\begin{lemma}
\label{lm:distinctpoles}
Let $S$ be a finite set of places in $\calR$ and $\alpha\in \calR$.  Then there are $a_1,a_2\in k$ with $a_2\neq 0$ such that
$$
    \supp\left(\divs\left(\frac{\alpha}{a_1\alpha+a_2}\right)^{-}\right)\cap S =\emptyset.
$$
\end{lemma}
\begin{proof}
Set
$$
     M=\left\{ \pi_\frakP(\alpha) \mid  \mbox{$\forall\, \frakP\in S$ with $\nu_\frakP(\alpha)\geq 0$} \right\}. $$
Then $M$ is a finite set in $\overkt$. Let $a_1, a_2\in k$ satisfy that $a_2\neq 0$ and  $a_1c+a_2\neq 0$ for all $c\in M$.
For $\frakP\in S$ with $\nu_\frakP(\alpha)\geq 0$, one has that
$$
    \pi_\frakP(a_1\alpha+a_2)=a_1\pi_\frakP(\alpha)+a_2\neq 0,\, {\mathrm i.e.} \,\,\nu_\frakP(a_1\alpha+a_2)=0.
$$
This implies that $\nu_\frakP(\alpha/(a_1\alpha+a_2))=\nu_\frakP(\alpha)\geq 0$. On the other hand, for $\frakP\in S$ with $\nu_\frakP(\alpha)<0$, one has that
$$
    \nu_\frakP(\alpha/(a_1\alpha+a_2))=\nu_\frakP(\alpha)-\nu_\frakP(a_1\alpha+a_2)=\nu_\frakP(\alpha)-\nu_\frakP(\alpha)=0.
$$
In both cases, $\frakP$ is not a pole of $\alpha/(a_1\alpha+a_2)$. Thus $a_1,a_2$ satisfy the requirement.
\end{proof}
The main result of this section is the following theorem which is a special case of Lemma 2 of \cite{eremenko}. The original proof of Lemma 2 of \cite{eremenko} contains a small gap. We shall fill in this gap in the proof.
\begin{theorem}
\label{thm:boundpoints1}
Let $f$ be an irreducible polynomial in $\overkt[x_0,x_1]$ of degree $n_0$ with respect to $x_0$ and of degree $n_1$ with respect to $x_1$.  Suppose that $n=\tdeg(f)$ and $N\geq 1$.  Then for every $c_0,c_1\in \overkt$ satisfying $f(c_0,c_1)=0$, one has that
$$
    \left(1-\frac{n}{N+n}\right)n_0T(c_0)-C \leq n_1T(c_1) \leq \left(1+\frac{n}{N}\right)n_1T(c_1)+C
$$
where
\begin{equation}
\label{eq:boundsforpoints}
C=2^{s^2/2+15s/2+10}(2Nn+n^2+2^{s-2})^4 n^{s+9}(n+1)^4T(f)/N
\end{equation}
and $s$ is the number of quadratic transformations which are applied to resolve the singularities of $f=0$.
\end{theorem}
\begin{proof}
If one of $c_i$ is in $k$ then the height of the other one is not greater than $T(f)$. The inequalities then obviously hold. In the following, we assume that  neither $c_0$ nor $c_1$ is in $k$.

Let $\calR$ be the algebraic function field associated to $f$ and $\alpha,\beta\in \calR\setminus \overkt$ satisfy that $f(\alpha,\beta)=0$.
Choose $a_1,a_2\in k$ such that $a_2\neq 0$ and
$$\supp(\divs(\alpha/(a_1\alpha+a_2))^{-})\cap \divs(\beta)^{-}=\emptyset.$$
Such $a_1,a_2$ exist due to Lemma~\ref{lm:distinctpoles}. Set $\bar{\alpha}=\alpha/(a_1\alpha+a_2)$. Consider the divisor
$$D=(N+n)n_1\divs(\beta)^{-} -Nn_0\divs(\bar{\alpha})^{-}.$$
Note that $\deg(\divs(\bar{\alpha})^{-})=n_1$, $\deg(\divs(\beta)^{-})=n_0$ and
$$n_0n_1\geq n_0+n_1-1\geq n-1. $$
 So
 $$\deg(D)=nn_0n_1\geq n(n-1). $$
 This implies that $\deg(D)$ is greater than the genus of $f=0$ and thus $\frakL(D)\neq \{0\}$. Denote $\bfxi=(\alpha,\beta,1)$ and by $F(x_0,x_1,x_2)$ the homogenization of $f$ . We claim that $T_\bfxi(D)\leq T(f)$. Note that $T_\bfxi(D)=\max\{T_\bfxi(\divs(\bar{\alpha})^{-}), T_\bfxi(\divs(\beta)^{-})\}.$ For each $\bfa\in \calS_\bfxi(\divs(\beta)^{-})$, $\bfa$ is of the form $(b_0,b_1,0)$ where $b_0,b_1$ satisfies that $F(b_0,b_1,0)=0$. So $T(\bfa)\leq T(F)=T(f)$ and then $T_\bfxi(\divs(\beta)^{-})\leq T(f)$.  If $a_1=0$ then each point in $\calS_\bfxi(\divs(\bar{\alpha})^{-})$ is of the form $(b_0,b_1,0)$ too and so $T_\bfxi(\divs(\bar{\alpha})^{-})\leq T(f)$. Otherwise, for each $\frakP\in \divs(\bar{\alpha})^{-}$, one has that
$$\nu_\frakP(\alpha)=\nu_\frakP(a_2\bar{\alpha}/(1-a_1\bar{\alpha}))=\nu_\frakP(\bar{\alpha})-\nu_\frakP(1-a_1\bar{\alpha})=0.$$
Moreover $\pi_\frakP(\alpha)=-a_2/a_1$. This implies that each point of $\calS_\bfxi(\divs(\bar{\alpha})^{-})$ is of the form $(-a_2/a_1, b, 1)$ whose height is not greater than $T(f)$. Thus $T_\bfxi(\divs(\bar{\alpha})^{-})\leq T(f)$.  Our claim is proved. Note that
$$\deg(D^{+}+D^{-})=2Nn_0n_1+n n_0n_1\leq (2N+n)n^2.$$
Suppose that $\gamma\in \frakL(D)\setminus\{0\}$. Due to Theorem~\ref{thm:riemann-roch2}, $\gamma=G(\bfxi)/H(\bfxi)$ where $G,H$ are two homogeneous polynomials of degree not greater than 
$$
  2^{2s+1}(n+1)\left(\deg(D^{+}+D^{-})+2^{s-2}n\right)\leq 2^{2s+1}n(n+1)(2Nn+n^2+2^{s-2})
$$
and 
\begin{align*}
   T(G), T(H) &\leq 2^{s^2/2+15s/2+8}n^{s+6}(n+1)^3\left(\deg(D^{+}+D^{-})+2^{s-2}n\right)^3 T(f)\\
   &\leq 2^{s^2/2+15s/2+8}n^{s+9}(n+1)^3\left(2Nn+n^2+2^{s-2}\right)^3 T(f).
\end{align*}
Set 
$$
   \tilde{C}=2^{s^2/2+15s/2+9} n^{s+9}(n+1)^4(2Nn+n^2+2^{s-2})^4T(f).   
$$ 
Without loss of generality, we assume that $G(x_0,x_1,1)$ and $H(x_0,x_1,1)$ have no common factor. Otherwise, by Corollary~\ref{cor:factor}, we may replace $G$ and $H$ by $G/W$ and $H/W$ where $W$ is the greatest common factor of $G$ and $H$. Moreover, multiplying by suitable elements in $\overkt$ if necessary, we can assume that both $G(x_0,x_1,1)$ and $H(x_0,x_1,1)$ have 1 as a coefficient.
Let $\frakP$ be a place of $\calR$ containing $\alpha-c_0$ and $\beta-c_1$. Then $\nu_\frakP(\alpha)=0$ and $\nu_\frakP(\beta)=0$. As $\gamma\in \frakL(D)$, $\nu_\frakP(\gamma)\geq 0$. If $\nu_\frakP(\gamma)>0$, then $\nu_\frakP(G(\alpha,\beta,1))>0$ and so $G(c_0,c_1,1)=0$. Consequently, $(c_0,c_1)$ is a common point of $G(x_0,x_1,1)=0$ and $f(x_0,x_1)=0$. Proposition~\ref{prop:intersection} implies that $T(c_i)\leq \deg(G)T(f)+nT(G)$. It is easy to verify that in this case $T(c_0), T(c_1)$ satisfy the required inequalities. Therefore we only need to prove the case $\nu_\frakP(\gamma)=0$.

Set
\begin{align*}
  h_1(x_1,y)&=\res_{x_0}(f(x_0,x_1),H(x_0,x_1,1)y-G(x_0,x_1,1)\\
  h_2(x_2,y)&=\res_{x_1}(h_1(x_1,y),  x_2-x_1^{(N+n)n_1})
\end{align*}
 where $\res_{x_0}(f,g)$ denotes the resultant of $f$ and $g$ with respect to $x_0$. Note that $h_1(x_1,y)\neq 0$, because $G(x_0,x_1,1)$ and $H(x_0,x_1,1)$ have no common factor. As $D$ is not effective, $\gamma\notin \overkt$. Furthermore, as $h_1(x_1,\gamma)=0$, $\deg(h_1,x_1)>0$. It is easy to see that $h_2\neq 0$ and $h_2(\beta^{(N+n)n_1},\gamma)=0$. Let $\tilde{h}_2$ be an irreducible factor of $h_2$ in $\overkt[x_2,y]$ such that $\tilde{h}_2(\beta^{(N+n)n_1},\gamma)=0$.
Propositions~\ref{prop:resultant} and~\ref{prop:height2}  imply that
\begin{align*}
   T(\tilde{h}_2)&\leq T(h_2)\leq  (N+n)n_1 T(h_1)\\
   & \leq (N+n)n_1\left(\deg(H)T(f)+n(T(G)+T(H))\right)\\
   &\leq (N+n)n_1(2^{2s+1}n(n+1)(2Nn+n^2+2^{s-2})T(f) \\
     &   +2n2^{s^2/2+15s/2+8}n^{s+9}(n+1)^3\left(2Nn+n^2+2^{s-2}\right)^3 T(f) )\\
   &\leq (N+n)n_1 2^{s^2/2+15s/2+9}n^{s+9}(n+1)^4\left(2Nn+n^2+2^{s-2}\right)^3 T(f)\\
   &\leq  2^{s^2/2+15s/2+9}n^{s+9}(n+1)^4\left(2Nn+n^2+2^{s-2}\right)^4 T(f)= \tilde{C}.
\end{align*}
Remark that
$$
    \divs(\beta^{(N+n)n_1})^{-}=(N+n)n_1\divs(\beta)^{-}\geq D^{+}\geq \divs(\gamma)^{-}.
$$
Note that $\pi_\frakP(\beta)=c_1$.
By Lemma~\ref{lm:pointsinequality},
\begin{align}
\label{eq:betagamma}
  T(\pi_\frakP(\gamma)) &\leq T\left(\pi_\frakP(\beta)^{(N+n)n_1}\right)+T(\tilde{h}_2) \\ \notag
  & \leq (N+n) n_1T\left(c_1\right)+\tilde{C}.
\end{align}
Similarly, let $r_1(x_0,y)=\res_{x_1}(f(x_0,x_1), G(x_0,x_1,1)y-H(x_0,x_1,1))$ and
$$r_2(x_2,y)=\res_{x_0}\left(r_1(x_0,y), (a_1x_0+a_2)^{Nn_0}x_2-x_0^{Nn_0}\right).$$ Then $r_2\neq 0$ and $r_2(\bar{\alpha}^{Nn_0},\gamma^{-1})=0$. Let $\tilde{r}_2$ be an irreducible factor of $r_2$ in $\overkt[x_2,y]$ such that $\tilde{r}_2(\bar{\alpha}^{Nn_0},\gamma^{-1})=0$. Applying Propositions~\ref{prop:resultant} and~\ref{prop:height2} again yields that
\begin{align*}
T(\tilde{r})\leq  T(r_2)&\leq Nn_0T(r_1)\leq Nn_0\left( \deg(G)T(f)+n(T(H)+T(G))\right).
\end{align*}
An argument similar to the above implies that $T(\tilde{r})\leq \tilde{C}$.
Since $\supp(\divs(\bar{\alpha})^{-})\cap \supp(\divs(\beta)^{-})=\emptyset$, one has that $Nn_0\delta(\bar{\alpha})^{-}=D^{-}$ and thus
$$
     \divs(\bar{\alpha}^{Nn_0})^{-}=Nn_0\divs(\bar{\alpha})^{-}=D^{-}\leq \divs(\gamma)^{+}=\divs(\gamma^{-1})^{-}.
$$
Furthermore since $a_1c_0+a_2\neq 0$, $\pi_\frakP(\bar{\alpha})=c_0/(a_1c_0+a_2)$.
By Lemma~\ref{lm:pointsinequality},
\begin{align*}
\label{eq:alphagamma}
Nn_0 T\left(\frac{c_0}{a_1c_0+a_2}\right) &=T(\pi_\frakP(\bar{\alpha})^{Nn_0})\leq T(\pi_\frakP(\gamma^{-1}))+T(\tilde{r}_2)\\
&=T(\pi_\frakP(\gamma))+T(\tilde{r}_2) \leq  T(\pi_\frakP(\gamma))+\tilde{C}.
\end{align*}
which together with (\ref{eq:betagamma}) gives
$$
  N n_0T\left(\frac{c_0}{a_1c_0+a_2}\right)\leq (N+n) n_1T\left(c_1\right)+2\tilde{C}.
$$
Proposition~\ref{prop:heightproperty} implies that
\begin{align*}
  \left(1-\frac{n}{N+n}\right)n_0T\left(\frac{c_0}{a_1c_0+a_2}\right)-\frac{2\tilde{C}}{N+n} &\leq  \left(1-\frac{n}{N+n}\right)n_0T\left(c_0\right)-\frac{2\tilde{C}}{N+n}\\
    &\leq n_1T(c_1).
\end{align*}
To prove the inequality in the opppsite direction, consider
$$\tilde{D}=(N+n)n_0\divs(\bar{\alpha})^{-}-Nn_1\divs(\beta)^{-}.$$
Remark that $\deg(D^{+}+D^{-})=\deg(\tilde{D}^{+}+\tilde{D}^{-})$ and $T_\bfxi(D)=T_\bfxi(\tilde{D})$. We have the same bounds for elements in $\calL(\tilde{D})$. A similar argument then implies that
$$
    n_1T(c_1)\leq  \frac{N+n}{N}n_0T(c_0)+\frac{2\tilde{C}}{N}\leq \left(1+\frac{n}{N}\right)n_1T(c_0)+\frac{2\tilde{C}}{N}.
$$
Set $C=2\tilde{C}/N$. Then one gets the required inequalities.
\end{proof}

\section{Main results}
In this section, we always ssume that $f(y,y')=\sum_{i=0}^d a_i(y)y'$ is irreducible over $k(t)$ and
$$
       \ell=\msindex(f)=\max_{i=0}^d \{ \deg(a_i)-2(d-i)\}>0.
$$
Pick $c\in k$ such that $a_0(c)\neq 0$.
Set $y=(cz+1)/z$.  Then $y'=-z'/z^2$. Set
$$
     b_i(z)=a_i((cz+1)/z)z^{\ell+2d-2i} (-1)^i
$$
where $i=0,\dots,d$. Then an easy calculation yields that
\begin{align*}
   g(z,z')=\sum_{i=0}^d b_i(z)z'^i=z^{2d+\ell}f\left(\frac{cz+1}{z}, \frac{-z'}{z^2}\right).
\end{align*}
 As $a_0(c)\neq 0$,
$
   \deg(a_0(cz+1)/z)z^{\deg(a_0)})=\deg(a_0).
$
This implies that $$\deg(b_0)=\ell+2d>2d.$$
Then $\tdeg(g)=2d+\ell$ because
$$
  2d+\ell\leq \tdeg(g)=\max\{ \deg(b_i)+i\} \leq \max\{2d+\ell-2i\}=2d+\ell.
$$
We claim that $g(z,z')$ is irreducible over $k(t)$. First of all, assume that $\ell=\deg(a_{i_0})-2(d-i_0)$ for some $0
\leq i_0 \leq d$. Then we have that
$$
    b_{i_0}(0)=(-1)^{i_0}\cdot\mbox{the leading coefficient of $a_{i_0}(y)$}\neq 0.
$$
If $\gcd(b_0,\dots,b_d)\neq 1$ then the $b_i(z)$ have common zeroes and none of common zeroes is zero. It is easy to see that $(c\eta+1)/\eta$ is a common zero of all $a_i(y)$ if $\eta$ is a common zero of all $b_i$. This contradicts with the fact that $\gcd(a_0,\dots,a_d)=1$.  Secondly, if $g(z,z')$ has a factor with positive degree in $z'$ then $f(y,y')$ will have a factor with positive degree in $y'$, a contradiction. This proves our claim. Remark that $r(t)$ is a nontrivial rational solution of $g(z,z')=0$ if and only if $(cr(t)+1)/r(t)$ is a nontrivial rational solution of $f(y,y')=0$.  The main result of this paper is the following theorem.
\begin{theorem}
\label{thm:boundforsols}
Assume that $f(y,y')=0$ is a first order AODE with positive $\msindex$ and assume further that $f(y,y')$ is irreducible over $k(t)$. Then if $r(t)$ is a rational solution of $f(y,y')=0$ then
$$
   \deg(r(t))\leq (54n^3+9n^2+2^{5n^2})^4 n^{5n^2+12}2^{11n^4+43n^2+34}T(f).
$$
where $n=\tdeg(f)$.
\end{theorem}
\begin{proof}
We shall use the notations as above. Due to the above discussion, we only need to consider the differential equation $g(z,z')=0$. Denote $n=\tdeg(f)$ and $d=\deg(f,y')$. One sees that $T(g)\leq T(f), d=\deg(g,z')$ and
$$
  \deg(g,z)=2d+\ell=\tdeg(g)\leq 3n.
 $$
 Suppose that
$$
   g=h_1h_1\dots h_m
$$
where $h_i$ is irreducible over $\overkt$. Since $g$ is irreducible over $k(t)$, one has that all $h_i$ are conjugate to each other and then
\begin{align*}
       \deg(h_i, z)=\deg(g,z)/m,\,\,\deg(h_i, z')=\deg(g,z')/m=d/m.
\end{align*}
By Corollary~\ref{cor:factor}, $T(h_i)\leq T(g)\leq T(f)$. Assume that $r(t)$ is a rational solution of $g(z,z')=0$ then $r(t)$ is a rational solution of all $h_i=0$. In particular, $h_1(r(t),r'(t))=0$. Denote $\tilde{n}=\tdeg(h_1)$ and $\tilde{d}=\deg(h_1,z')$. Then
\begin{align*}
    \tilde{n}& =\tdeg(g)/m=(2d+\ell)/m\leq 3n/m \\
    \tilde{d}&=\deg(g,z')/m=d/m.
\end{align*}
Set $N=\tilde{n}^2$. By Theorem~\ref{thm:boundpoints1} and Remark~\ref{rem:heights}, one has that
$$
\frac{N}{N+\tilde{n}}\frac{\deg(g,z)}{\deg(g,z')} \deg(r(t))-\frac{C}{\tilde{d}}\leq \deg(r'(t))
$$
where  
\begin{align*}
C=(2\tilde{n}^3+\tilde{n}^2+2^{s-2})^4\tilde{n}^{s+7}(\tilde{n}+1)^4 2^{\frac{s^2}{2}+\frac{15s}{2}+10}T(f).
\end{align*}
Note that $s$ is the number of quadratic transformations applied to transfer $h_1=0$ to an algebraic curve with only ordinary singularities. Due to Theorem 2 in Chapter 7 of \cite{fulton}, $s$ can be chosen to be an integer not greater than 
$$(\tilde{n}-1)(\tilde{n}-2)/2\leq (3n-1)(3n-2)/2\leq 9n^2/2.$$
Remark that $\deg(r'(t))\leq 2\deg(r(t))$. Thus
\begin{equation}
\label{eq:leftinequality}
    \left(\frac{N}{N+\tilde{n}}\frac{\deg(g,z)}{\deg(g,z')}-2\right)\deg(r(t))\leq \frac{C}{\tilde{d}}.
\end{equation}
As $m$ divides both $\deg(g,z)$ and $\deg(g,z')$, $m$ divides $\ell$. Set $\ell=m\bar{\ell}$. Then
\begin{align*}
    \frac{N}{N+\tilde{n}}\frac{\deg(g,z)}{\deg(g,z')}-2=\frac{\tilde{n}^2(2d+\ell)}{(\tilde{n}^2+\tilde{n})d}-2
    =\frac{\tilde{n}\ell-2d}{(\tilde{n}+1)d}&\geq \frac{\bar{\ell}\deg(g,z)-2d}{(\tilde{n}+1)d}\\
    &\geq \frac{1}{(\tilde{n}+1)d}.
\end{align*}
This together with (\ref{eq:leftinequality}) implies that 
\begin{align*}
    \deg(r(t))&\leq m(\tilde{n}+1) C \\
    &\leq 4n(2\tilde{n}^3+\tilde{n}^2+2^{s-2})^4\tilde{n}^{s+7}(\tilde{n}+1)^42^{\frac{s^2}{2}+\frac{15s}{2}+10}T(f)\\
    &\leq n(54n^3+9n^2+2^{s-2})^4 (4n)^{s+11}2^{\frac{s^2}{2}+\frac{15s}{2}+12}T(f)\\
    &\leq (54n^3+9n^2+2^{s-2})^4 n^{s+12} 2^{\frac{s^2}{2}+\frac{19s}{2}+34}T(f)\\
    &\leq (54n^3+9n^2+2^{5n^2})^4 n^{5n^2+12}2^{11n^4+43n^2+34}T(f).
\end{align*}
The second inequality holds because $ m(\tilde{n}+1)\leq 3n+m\leq 4n$.
\end{proof}
\begin{remark}
Theorem~\ref{thm:boundforsols} implies that an autonomous first order AODE $f=0$ with positive $\msindex$ has no nontrival rational solutions, because $T(f)=0$. In fact, suppose that $f=0$ has a nontrival rational solution. Then it will have infinitely many rational solutions. By Corollary 4.6 of \cite{feng-feng}, $f=0$ has no movable singularities. However, as $f=0$ has positive $\msindex$, Fuchs' criterion implies that $f=0$ has movable singularities, a contradiction.
\end{remark}

In \cite{vo-grasegger-winkler1}, the authors developed two algorithms to compute rational solutions of maximally comparable first-order AODEs and first order quasi-linear AODEs respectively. Let us first recall the definition of maximally comparable first order AODEs. Suppose that $f=\sum_{i,j}a_{i,j}y^iy'^j$ is a differential polynomial over $k(t)$. Denote
$$
S(f)=\{(i,j)\in \N^2 \mid a_{i,j}\neq 0\}.
$$
If there is $(i_0,j_0)\in S(f)$ satisfying that $i_0+j_0\geq i+j$ and $i_0+2j_0>i+2j$ for every $(i,j)\in S(f)$, then we say that $f$ is maximally comparable. The following examples shows that their algorithms can not deal with all first order AODEs with positive $\msindex$.
\begin{example}
Let
$$
f=yy'^m+y^{2m+1}+t
$$
where $m\geq 1$. Then $S(f)=\{(1,m),(2m+1,0),(0,0)\}$. Since $2m+1+0\geq 1+m$ but $2m+1+2\cdot 0=1+2\cdot m$, $f$ is not maximally comparable. 
On the other hand, we have that 
$
   \msindex(f)=1>0.
$
\end{example}

\section*{References}
\bibliographystyle{plain}
\bibliography{symsolutions}

\begin{thebibliography}{10}

\bibitem{aroca-cano-feng-gao}
J.~M. Aroca, J.~Cano, R.~Feng, and X.~S. Gao.
\newblock Algebraic general solutions of algebraic ordinary differential
  equations.
\newblock In {\em I{SSAC}'05}, pages 29--36. ACM, New York, 2005.

\bibitem{barkatou}
Moulay~A. Barkatou.
\newblock On rational solutions of systems of linear differential equations.
\newblock volume~28, pages 547--567. 1999.
\newblock Differential algebra and differential equations.

\bibitem{chau-winkler}
L.~X. Ch\^{a}u~Ng\^{o} and Franz Winkler.
\newblock Rational general solutions of first order non-autonomous
  parametrizable {ODE}s.
\newblock {\em J. Symbolic Comput.}, 45(12):1426--1441, 2010.

\bibitem{chevalley}
Claude Chevalley.
\newblock {\em Introduction to the theory of algebraic functions of one
  variable}.
\newblock Mathematical Surveys, No. VI. American Mathematical Society,
  Providence, R.I., 1963.

\bibitem{eremenko}
A.~Eremenko.
\newblock Rational solutions of first-order differential equations.
\newblock {\em Ann. Acad. Sci. Fenn. Math.}, 23(1):181--190, 1998.

\bibitem{feng-gao}
Ruyong Feng and Xiao-Shan Gao.
\newblock A polynomial time algorithm for finding rational general solutions of
  first order autonomous {ODE}s.
\newblock {\em J. Symbolic Comput.}, 41(7):739--762, 2006.

\bibitem{feng-feng}
Shuang Feng and Ruyong Feng.
\newblock Descent of ordinary differential equations with rational general
  solutions.
\newblock {\em J. Syst. Sci. Complex.}, 2020.

\bibitem{freitag-moosa}
James Freitag and Rahim Moosa.
\newblock Finiteness theorems on hypersurfaces in partial
  differential-algebraic geometry.
\newblock {\em Adv. Math.}, 314:726--755, 2017.

\bibitem{fulton}
William Fulton.
\newblock {\em Algebraic curves}.
\newblock Advanced Book Classics. Addison-Wesley Publishing Company, Advanced
  Book Program, Redwood City, CA, 1989.
\newblock An introduction to algebraic geometry, Notes written with the
  collaboration of Richard Weiss, Reprint of 1969 original.

\bibitem{hess}
F.~Hess.
\newblock Computing {R}iemann-{R}och spaces in algebraic function fields and
  related topics.
\newblock {\em J. Symbolic Comput.}, 33(4):425--445, 2002.

\bibitem{huang-ierardi}
Ming-Deh Huang and Doug Ierardi.
\newblock Efficient algorithms for the {R}iemann-{R}och problem and for
  addition in the {J}acobian of a curve.
\newblock {\em J. Symbolic Comput.}, 18(6):519--539, 1994.

\bibitem{kovacic}
Jerald~J. Kovacic.
\newblock An algorithm for solving second order linear homogeneous differential
  equations.
\newblock {\em J. Symbolic Comput.}, 2(1):3--43, 1986.

\bibitem{lang}
Serge Lang.
\newblock {\em Fundamentals of {D}iophantine geometry}.
\newblock Springer-Verlag, New York, 1983.

\bibitem{matsuda}
Michihiko Matsuda.
\newblock {\em First-order algebraic differential equations}, volume 804 of
  {\em Lecture Notes in Mathematics}.
\newblock Springer, Berlin, 1980.
\newblock A differential algebraic approach.

\bibitem{ritt}
Joseph~Fels Ritt.
\newblock {\em Differential algebra}.
\newblock Dover Publications, Inc., New York, 1966.

\bibitem{sendra-winkler}
J.~Rafael Sendra and Franz Winkler.
\newblock Tracing index of rational curve parametrizations.
\newblock {\em Comput. Aided Geom. Design}, 18(8):771--795, 2001.

\bibitem{singer1}
Michael~F. Singer.
\newblock Liouvillian solutions of {$n$}th order homogeneous linear
  differential equations.
\newblock {\em Amer. J. Math.}, 103(4):661--682, 1981.

\bibitem{vanderput-singer}
Marius van~der Put and Michael~F. Singer.
\newblock {\em Galois theory of linear differential equations}, volume 328 of
  {\em Grundlehren der Mathematischen Wissenschaften [Fundamental Principles of
  Mathematical Sciences]}.
\newblock Springer-Verlag, Berlin, 2003.

\bibitem{vanhoeij-ragot-ulmer-weil}
Mark van Hoeij, Jean-Fran\c{c}ois Ragot, Felix Ulmer, and Jacques-Arthur Weil.
\newblock Liouvillian solutions of linear differential equations of order three
  and higher.
\newblock volume~28, pages 589--609. 1999.
\newblock Differential algebra and differential equations.

\bibitem{vo-grasegger-winkler2}
N.~Thieu Vo, Georg Grasegger, and Franz Winkler.
\newblock Deciding the existence of rational general solutions for first-order
  algebraic {ODE}s.
\newblock {\em J. Symbolic Comput.}, 87:127--139, 2018.

\bibitem{vo-grasegger-winkler1}
Thieu~N. Vo, Georg Grasegger, and Franz Winkler.
\newblock Computation of all rational solutions of first-order algebraic
  {ODE}s.
\newblock {\em Adv. in Appl. Math.}, 98:1--24, 2018.

\bibitem{winkler}
Franz Winkler.
\newblock The algebro-geometric method for solving algebraic differential
  equations---a survey.
\newblock {\em J. Syst. Sci. Complex.}, 32(1):256--270, 2019.

\end{thebibliography}
\end{document}